\documentclass[a4paper,11pt]{article}

\usepackage{classeJB}

\usepackage{arydshln}

\newcommand*\unit[1]{\bigl[\, \mathsf{#1} \,\bigr]}
\newcommand*\unitt[2]{\bigl[\, \mathsf{#1}\,/\,(\mathsf{#2}) \,\bigr]}

\newcommand{\Bi}{\mathrm{Bi}}
\newcommand{\Fo}{\mathrm{Fo}}

\newcommand{\Tref}{T_{\,\reff}}
\newcommand{\tref}{t_{\,\reff}}

\newcommand{\apr}{\mathrm{apr}}
\newcommand{\fin}{\mathrm{f}}
\newcommand{\obs}{\mathrm{obs}}
\newcommand{\dir}{\mathrm{dir}}
\newcommand{\real}{\mathrm{r}}
\newcommand{\reff}{\mathrm{ref}}

\newtheorem{proposition}{Proposition}

\newcommand{\DF}{\textsc{Du$\,$Fort}--\textsc{Frankel}}

\title{
\vspace{-1.5cm}
Evaluation of the reliability of building energy performance models for parameter estimation }

\author{Julien Berger\textsuperscript{a}$^{\ast}$, Denys Dutykh\textsuperscript{b}\\
\date{\today\vspace{-0.5cm}}}

\begin{document}

\maketitle

\begin{center}
\small
\textsuperscript{a} Univ. Grenoble Alpes, Univ. Savoie Mont Blanc, CNRS, LOCIE,
73000 Chambéry, France \\~
\textsuperscript{b} Univ. Grenoble Alpes, Univ. Savoie Mont Blanc, CNRS, LAMA,	
73000 Chambéry, France \\
$^{\ast}$corresponding author, e-mail address : julien.berger@univ-smb.fr\\
\end{center}

\begin{abstract}

The fidelity of a model relies both on its accuracy to predict the physical phenomena and its capability to estimate unknown parameters using observations. This article focuses on this second aspect by analyzing the reliability of two mathematical models proposed in the literature for the simulation of heat losses through building walls. The first one, named DF, is the classical heat diffusion equation combined with the \DF ~numerical scheme. The second is the so-called RC lumped approach, based on a simple ordinary differential equation to compute the temperature within the wall. The reliability is evaluated following a two stages method. First, samples of observations are generated using a pseudo-spectral numerical model for the heat diffusion equation with known input parameters. The results are then modified by adding a noise to simulate experimental measurements. Then, for each sample of observation, the parameter estimation problem is solved using one of the two mathematical models. The reliability is assessed based on the accuracy of the approach to recover the unknown parameter. Three case studies are considered for the estimation of (\textit{i}) the heat capacity, (\textit{ii}) the thermal conductivity or (\textit{iii}) the heat transfer coefficient at the interface between the wall and the ambient air. For all cases, the DF mathematical model has a very satisfactory reliability to estimate the unknown parameters without any bias. However, the RC model lacks of fidelity and reliability. The error on the estimated parameter can reach $40\%$ for the heat capacity, $80\%$ for the thermal conductivity and $450\%$ for the heat transfer coefficient. \\


\end{abstract}

\section{Introduction}

Within the environmental context, several works have been carried out to propose tools to assess the building energy performance. Among all physical phenomena involved, these tools are based on models to assess the heat losses through building walls. 

As illustrated in Figure~\ref{fig:building_model}, several steps can be identified in the construction of a model for the prediction of heat losses through walls. First, the model is based on a qualitative representation of the \emph{real physical} world. One can easily observe that in winter time, the heat flux is directed from the inside to the outside part of the wall. Then, this knowledge is translated into the so-called \emph{mathematical model}\footnote{The word ``mathematical'' is used because the mathematical language is used to write the model.}. The mathematical model includes the governing equations of the physical phenomena. The third step aims at building a \emph{numerical model}\footnote{The word ``numerical'' is adopted in the sense of computational.} to obtain a solution of the governing equations.
This model can employ numerical\footnote{Here, the word ``numerical'' stands for the type of method used to compute the solution.} or analytical, \emph{i.e.} approximate or exact, methods with defined discretisation of the continuous variables. Last, once the model is built, the numerical model can be developed using computational technologies so that predictions and analysis of the physical phenomena can be carried out. An alternative application is to estimate uncertain parameters entering in the definition of the model using experimental observations. Figure~\ref{fig:building_model} also illustrates the approximations introduced modeling procedure. Namely, some physical approximations are added when defining the mathematical model. Then, some numerical approximations appear when building the computational tool to solve the problem. 

Undeniably, the main issue is to build \emph{efficient} models. The word efficiency can designate several aspects. One important is to validate the numerical model by comparison to reference solutions. This work intends to check the numerical approximations introduced when obtaining the solution of the governing equations. A second aspect is to evaluate the fidelity or reliability of the model by comparison with experimental observations. To be more precise, the objective is to assess the physical approximations when translating the qualitative representation into the mathematical model. So the reliability of a model is its capacity to predict the physical phenomena. It is also the model's ability to estimate unknown parameters using experimental observations. Other criteria of efficiency can be based on computational costs, ease to implement, \emph{etc.} 

Nowadays, in building physics, two main mathematical models are employed in the literature to predict the heat transfer through building walls. The first one is the most known mathematical model, based on the heat diffusion equation proposed in the early work of \textsc{J. Fourier} in 1822 in \emph{Théorie analytique de la chaleur} \cite{Fourier_1822}. During the second world war, when no powerful computers were available, an analogical model was proposed to solve the heat diffusion equation as illustrated in Figure~\ref{fig:first_RC_model}. This ingenuous approach enabled fast computations to predict the heat transfer through walls caused by fire. Interested readers are invited to consult \cite{Kirkpatrick_1943,Lawson_1953,Robertson_1958}. Then, with the hardware evolution, numerical models have been proposed. Today, they are based on numerical approaches such as finite-differences or finite-volumes as surveyed in \cite{Mendes_2017}. The second mathematical model is the so-called RC approach. A lumped model for the heat diffusion equation is proposed based on ordinary differential equation \cite{Fraisse_2002,Kampf_2007,Naveros_2015}.

Despite the simplicity of these models, several works can be referenced in the literature using these two mathematical models to estimate uncertain parameters in building walls as for instance \cite{Roels_2017,Jimenez_2009} for the RC model or \cite{Berger_2016} for the heat transfer one. However, to our knowledge, no works have been proposed to evaluate the reliability of the two mathematical models. A complementary work \cite{Berger_2018b} investigates the fidelity of the two approaches to predict the physical phenomena with comparison to experimental observations. As a second step, this work intends to appraise their reliability to estimate unknown parameter from experimental observations through the resolution of parameter estimation problem \cite{Kabanikhin_2008,Kabanikhin_2011}. The article is organized as follows. First section presents the two direct mathematical models. Then, the procedure to evaluate the reliability for the estimation of unknown parameters is described. Particularly, samples of experimental observations are first generated using a pseudo-spectral numerical model for the heat equation. Then, for each generated sample, the parameter estimation problem is solved using two mathematical models. The metrics for evaluating the reliability are also proposed in Section~\ref{sec:evaluating_reliability}. In Section~\ref{sec:case_studies}, three case studies are considered for the estimation of (i) the heat capacity, (ii) the thermal conductivity and (iii) the heat transfer coefficient at the interface between the material and the ambient air. In last Section~\ref{sec:conclusion}, some general remarks are synthesized.
 
\begin{figure}
\centering
\includegraphics[width=.90\textwidth]{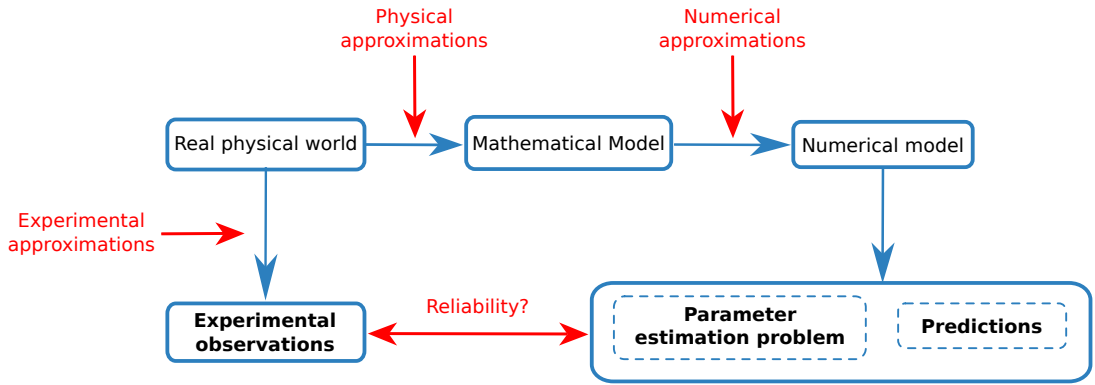}
\caption{Procedure of building models to represent the heat losses through building walls.}
\label{fig:building_model}
\end{figure}

\begin{figure}
\centering
\subfigure[\label{fig:RC_equivalence}]{\includegraphics[width=.45\textwidth]{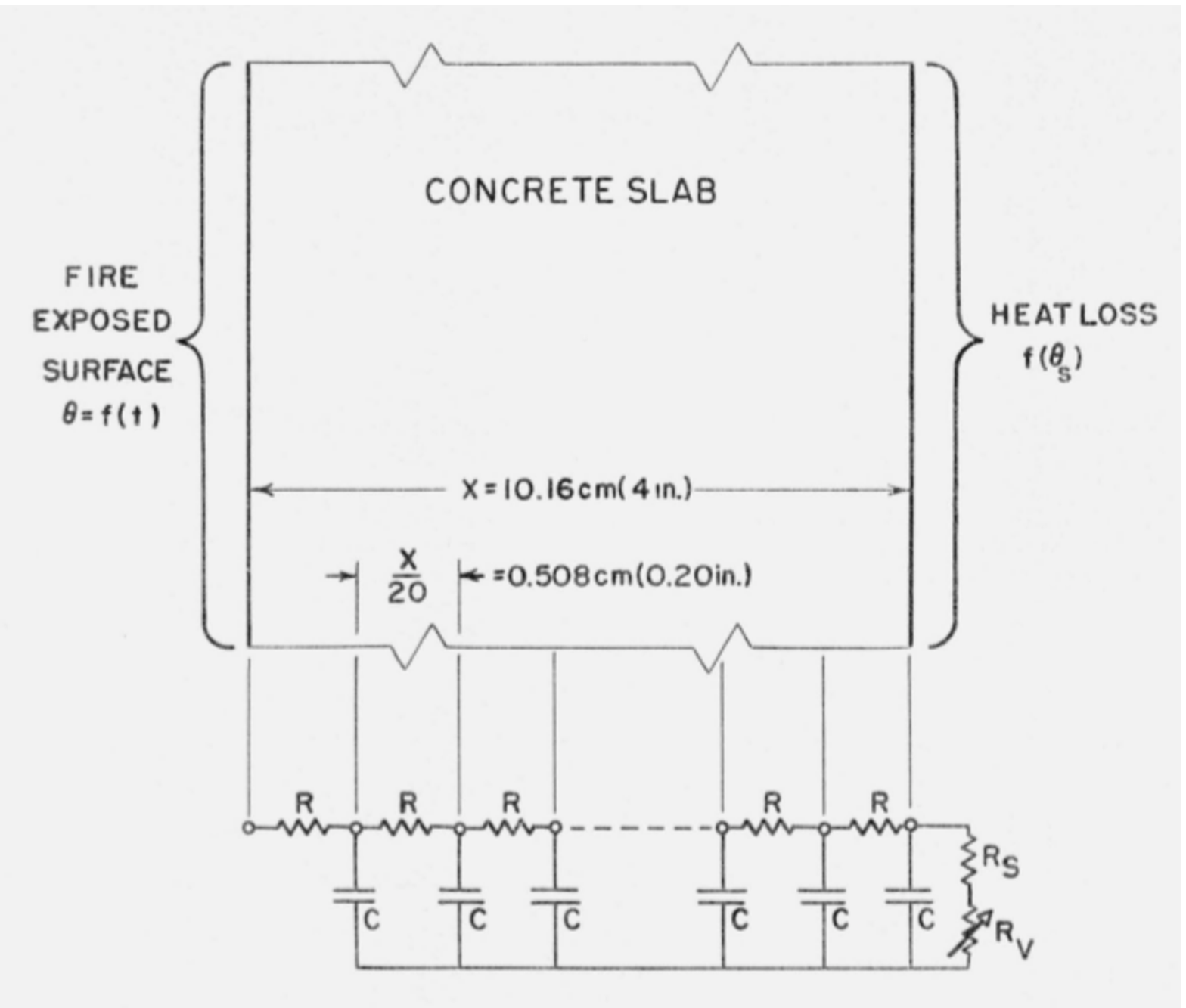}}  \hspace{0.2cm}
\subfigure[\label{fig:RC_analogique}]{\includegraphics[width=.45\textwidth]{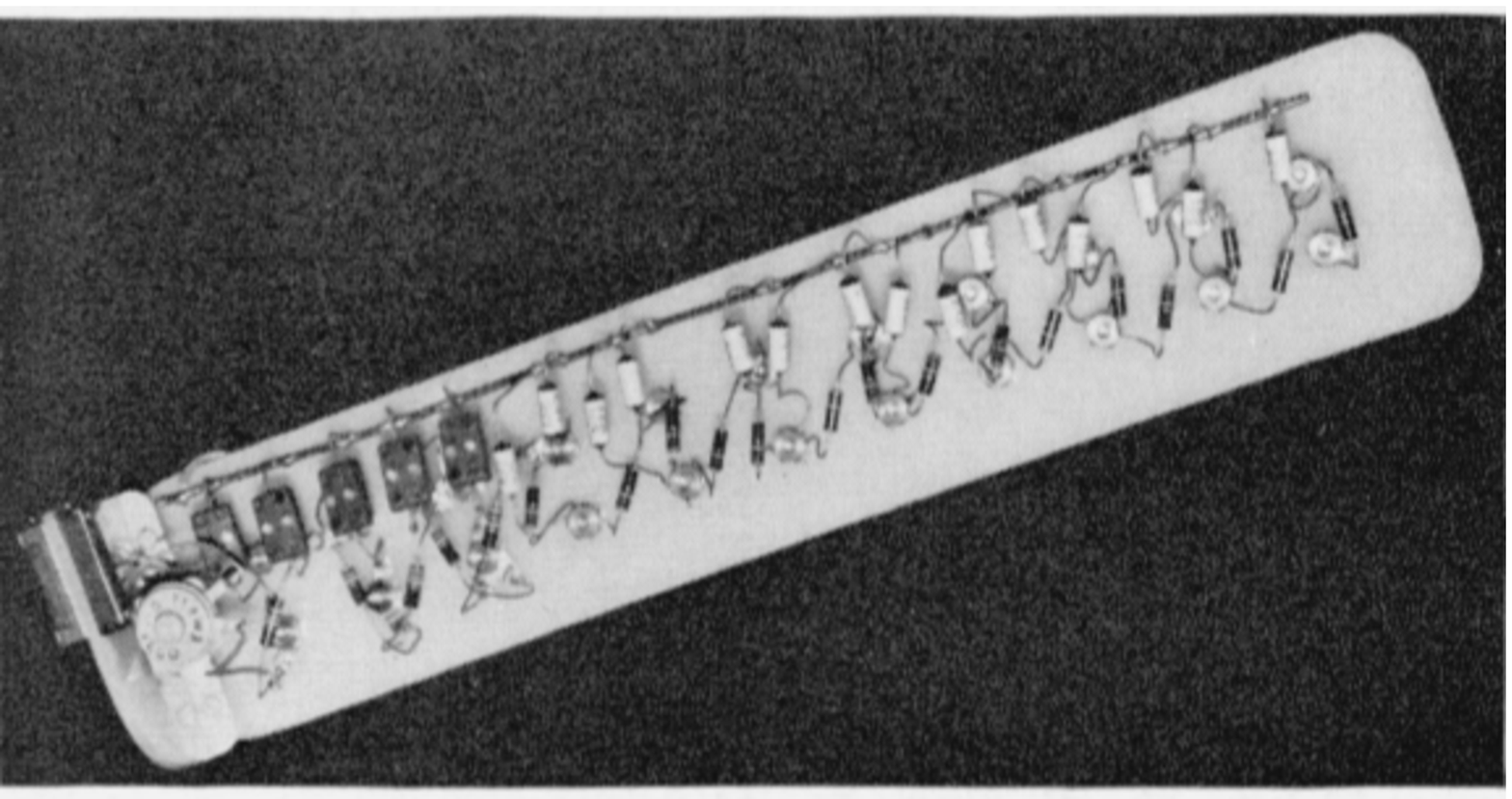}}
\caption{Illustration of the analogy between the physical problem of heat transfer through a wall with the electrical model \emph{(a)} and picture of the analogical model developed \emph{(b)}. Both illustrations are taken from \cite{Robertson_1958} with the authorization of the Journal of Research of NIST.}
\label{fig:first_RC_model}
\end{figure}

\section{Direct Mathematical models}
\label{sec:mathematical_model}

In this section, two models used to represent the physical phenomena of heat transfer in the wall are described.  Each model includes the \emph{mathematical model} translating the physical phenomena using a mathematical formalism. The \emph{numerical model} denotes the numerical method used to compute the solution of the mathematical model on a given space-time mesh. In the context of parameter estimation problem, the mathematical model is also called the \emph{direct model}. First, the heat diffusion model using the \DF ~numerical scheme is presented. This approach is denoted by DF in the whole manuscript. Then, the lumped RC model is described. 

\subsection{Heat diffusion with \DF ~numerical model}
\label{sec:DF_model}

\subsubsection{Heat diffusion in building material}

The field of interest is the temperature $T \ \unit{K}$ evolving in a building wall material illustrated in Figure~\ref{fig:domain}. The space domain is defined by $x \, \in \bigl[\, 0 \,,\, L \,\bigr]\,$, where $L \ \unit{m}$ is the length of the wall. The time domain is defined by $t \, \in \, \bigl[\, 0 \,,\, t_{\,\fin} \,\bigr] \,$. Thus, the function $T$ is defined by:
\begin{align*}
T \,:\, \bigl[\, 0 \,,\, L \,\bigr] \, \times \, \bigl[\, 0 \,,\, t_{\,\fin} \,\bigr] & \longrightarrow \, \mathbb{R}_{\,>\,0} \,, \\[4pt]
\bigl(\,x \,,\, t\,\bigr) & \longmapsto \, T(\,x\,,\,t\,) \,.
\end{align*}
The temperature verifies the diffusion equation:
\begin{align}
\label{eq:heat1d}
c \cdot \pd{T}{t} \egal k \cdot \pd{^{\,2} T}{x^{\,2}} \,,
\end{align}
where $k \ \unitt{W}{m \cdot K}$ is the thermal conductivity and $c \ \unitt{J}{kg \cdot K}$ is the volumetric heat capacity. 

At the boundaries, in the ambient air, time dependent temperatures $T_{\,\infty\,,\,L}\,(\,-\,)$  and $T_{\,\infty\,,\,R}\,(\,-\,)$ force the heat diffusion into the wall: 
\begin{align*}
T_{\,\infty\,,\,L\,,\,R} \,:\, \bigl[\, 0 \,,\, t_{\,\fin} \,\bigr] & \longrightarrow \, \mathbb{R}_{\,>\,0} \,, \\[4pt]
t & \longmapsto \, T_{\,\infty\,,\,L\,,\,R}\,\bigl(\,t\,\bigr) \,.
\end{align*}
\textsc{Robin}-type boundary conditions are assumed at the interface between ambient air and the wall,:
\begin{subequations}
\label{eq:BC}
\begin{align}
k \cdot \pd{T}{n} \egal h_{\,L} \cdot \Bigl(\, T \moins T_{\,\infty\,,\,L} \,\Bigr)\,, \qquad x \egal 0 \,, \qquad t \, \geqslant \, 0 \,, \label{eq:BC_L}\\[4pt]
k \cdot \pd{T}{n} \egal h_{\,R} \cdot \Bigl(\, T \moins T_{\,\infty\,,\,R} \,\Bigr)\,, \qquad x \egal L \,, \qquad t \, \geqslant \, 0 \,, \label{eq:BC_R}
\end{align}
\end{subequations}
where $h \ \unitt{W}{m^{\,2}\cdot K}$ is the surface heat transfer coefficient. The derivative is defined as $\displaystyle \pd{y}{n} \, \eqdef \, \pd{y}{x} \cdot n$ with $n \egal \pm \, 1\,$, depending on the considered boundary $\bigl\{\,0 \,,\, L \,\bigr\} \,$. 

A uniform temperature is assumed in the material as initial condition:
\begin{align*}
T \egal T_{\,0} \,, \qquad t \egal 0 \,.
\end{align*}
with $T_{\,0}\,(\,-\,)$ a constant function defined by:
\begin{align*}
T_{\,0} \,:\, \bigl[\, 0 \,,\, L \,\bigr] & \longrightarrow \, \mathbb{R}_{\,>\,0} \,, \\[4pt] 
x & \longmapsto \, T_{\,0} \,.
\end{align*}

\begin{figure}
\centering
\includegraphics[width=.45\textwidth]{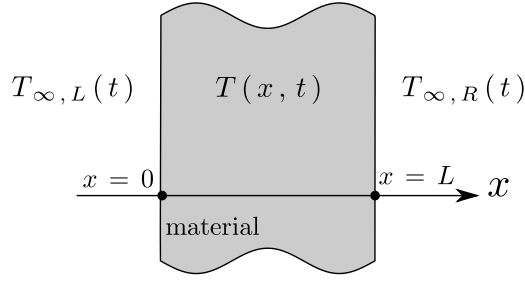}
\caption{Illustration of the physical domain.}
\label{fig:domain}
\end{figure}

\subsubsection{The \DF ~numerical model}

\paragraph{Dimensionless formulation}~\\

As discussed and thoroughly motivated in \cite{Berger_2018a,Nayfeh_2000,Kahan_1979}, it is of capital importance to obtain a dimensionless problem before elaborating a numerical model. For this, dimensionless fields are defined: 
\begin{align*}
u & \eqdef \, \frac{T}{\Tref}  \,,
&& u_{\,\infty\,,\,L} \, \eqdef \, \frac{T_{\,\infty\,,\,L}}{\Tref} \,,
&& u_{\,\infty\,,\,R} \, \eqdef \, \frac{T_{\,\infty\,,\,R}}{\Tref} \,,
&& u_{\,0} \, \eqdef \, \frac{T_{\,0}}{\Tref} \,,
\end{align*}
where $\Tref$ is user--defined reference temperature. The space and time coordinates are also transformed into dimensionless variables:
\begin{align*}
t^{\,\star} & \eqdef \, \frac{t}{\tref} \,,
&& x^{\,\star} \, \eqdef \, \frac{x}{L} \,.
\end{align*}
The thermophysical properties and the heat transfer coefficients are re-scaled using reference values: 
\begin{align*}
k^{\,\star} & \eqdef \, \frac{k}{k_{\,\reff}} \,,
&& c^{\,\star} \eqdef \, \frac{c}{c_{\,\reff}} \,,
&& h^{\,\star}_{\,L} \eqdef \, \frac{h_{\,L}}{h_{\,\reff}} \,,
&& h^{\,\star}_{\,R} \eqdef \, \frac{h_{\,R}}{h_{\,\reff}} \,,
\end{align*}
Last, dimensionless numbers are defined as the \textsc{Fourier} and \textsc{Biot} ones:
\begin{align*}
\Fo & \eqdef \, \frac{\tref \cdot k_{\,\reff}}{c_{\,\reff} \cdot L^{\,2}} \,,
&& \Bi \, \eqdef \, \frac{h_{\,\reff} \cdot L}{k_{\,\reff}} \,.
\end{align*}
The former quantifies the magnitude of diffusion inside the material while the second evaluates the importance of heat penetration from the ambient air to the solid part. With these transformations, the dimensionless problem is written as:
\begin{align}
\label{eq:heat_diffusion_dimless}
c^{\,\star} \cdot \pd{u}{t^{\,\star}} \egal \Fo \cdot k^{\,\star} \cdot \pd{^{\,2} u}{x^{\,\star \,2}} 
\end{align}
with the \textsc{Robin}-type boundary condition:
\begin{align}
\label{eq:BC_dimless}
k^{\,\star} \cdot \pd{u}{n} & \egal \Bi \cdot h^{\,\star}_{\,L} \cdot \Bigl(\, u \moins u_{\,\infty\,,\,L} \,\Bigr) \,, && x^{\,\star} \egal 0 \,, \qquad t^{\,\star} \, \geqslant \, 0 \\[4pt]
k^{\,\star} \cdot \pd{u}{n} & \egal - \, \Bi \cdot h^{\,\star}_{\,R} \cdot \Bigl(\, u \moins u_{\,\infty\,,\,R} \,\Bigr) \,, && x^{\,\star} \egal 1 \,, \qquad t^{\,\star} \, \geqslant \, 0  \,.
\end{align}
The initial condition is expressed as:
\begin{align*}
u\,\bigl(\,x^{\,\star}\,,\,0 \,\bigr) \egal u_{\,0} \,.
\end{align*}
To have a well-posed problem, initial and boundary conditions must be compatible. The dimensionless formulation is written in a way to highlight the parameter $k^{\,\star} \,$, $c^{\,\star}\,$ and $h^{\,\star}$ that will be estimated by solving the identification problem. In this work, the numerical values of the reference parameters are $t_{\,\reff} \egal 3600 \ \mathsf{s} \,$, $T_{\,\reff} \egal 273.15 \ \mathsf{K} \,$, $k_{\,\reff} \egal 1 \ \mathsf{W\,/\,(m \cdot K)} \,$, $c_{\,\reff} \egal 1.5 \ \mathsf{MJ\,/\,(m^{\,3}\cdot K)} $ and $h_{\,\reff} \egal 5 \ \mathsf{W\,/\,(m^{\,2} \cdot K)} \,$.

\paragraph{Numerical model}~\\

A uniform discretisation is considered for space and time intervals. The discretisation parameters are denoted using $\Delta t^{\,\star}$ for the time and $\Delta x^{\,\star}$ for the space. The discrete values of the function $u \, (\,x^{\,\star} \,,\,t^{\,\star}\,)$ are denoted by $u_{\,j}^{\,n} \ \eqdef \ u\,(\,x_{\,j}\,,\,t^{\,n}\,)$ with $j \egal 1 \,, \ldots \,, N_{\,x}$ and $n \egal 1 \,, \ldots \,, N_{\,t} \,$. It is important to note that the numerical model is built for the dimensionless problem using the \texttt{Matlab}\texttrademark environment.

The so--called \DF ~scheme is used. It is an explicit numerical scheme with an increased stability domain. Interested readers are invited to consult \cite{Du_Fort_1953} for the original work, \cite{Taylor_1970,Gasparin_2017a} for the results on the stability analysis and \cite{Gasparin_2017a,Gasparin_2017b} for further details and its applications for heat and moisture transfer in building porous materials. Since a complete description is provided in \cite{Gasparin_2017b}, only the main steps are recalled here. The idea of the approach is to replace the term $u_{\,j}^{\,n} \, \longleftarrow \, \displaystyle \frac{1}{2} \cdot \Bigl(\, u_{\,j}^{\,n+1} \plus u_{\,j}^{\,n-1} \,\Bigr)$ in the forward in time central scheme to obtain the following fully discrete dynamical system:
\begin{align}
\label{eq:DF}
u_{\,j}^{\,n+1}  \egal 
\frac{1}{1 \plus \lambda} \cdot
\Biggl(\, & \lambda \cdot u^{\,n}_{\,j+1}
\plus  \lambda \cdot u^{\,n}_{\,j-1}
\plus \Bigl(\, 1 \moins \lambda \, \Bigr) \cdot u^{\,n-1}_{\,j} \, \Biggr) \,,
\end{align}
with the coefficient:
\begin{align*}
\lambda \, \eqdef \, 2 \cdot \Fo \cdot \frac{k^{\,\star}}{c^{\,\star}} \cdot \frac{\Delta t^{\,\star}}{ \bigl(\, \Delta x^{\,\star} \,\bigr)^{\,2}} \,.
\end{align*}
It is important to remind that the boundary conditions are discretized using second order approach for the space derivatives to maintain the properties of stability \cite{Taylor_1970}. So, the boundary conditions~\eqref{eq:BC_dimless} are discretized according to:
\begin{align*}
\frac{k^{\,\star}}{2 \cdot \Delta x^{\,\star}} \cdot \Bigl(\, - \, u_{\,3}^{\,n} \plus 4 \cdot u_{\,2}^{\,n} \moins 3 \cdot u_{\,1}^{\,n}\,\Bigr)
& \egal \Bi_{\,L} \cdot h^{\,\star}_{\,L} \cdot \Bigl(\, u_{\,1}^{\,n} \moins u_{\,\infty\,,\,L} \,\Bigr) \,, \\[4pt]
\frac{k^{\,\star}}{2 \cdot \Delta x^{\,\star}} \cdot \Bigl(\, u_{\,N_{\,x}}^{\,n} \moins 4 \cdot u_{\,N_{\,x}-1}^{\,n} \plus 3 \cdot u_{\,N_{\,x}-2}^{\,n}\,\Bigr)
& \egal - \, \Bi_{\,R} \cdot h^{\,\star}_{\,R} \cdot \Bigl(\, u_{\,N_{\,x}}^{\,n} \moins u_{\,\infty\,,\,R} \,\Bigr) \,.
\end{align*}
The stencil of the \DF ~scheme is illustrated in Figure~\ref{fig:DF_stencil}.

\begin{figure}
\centering
\includegraphics[width=.45\textwidth]{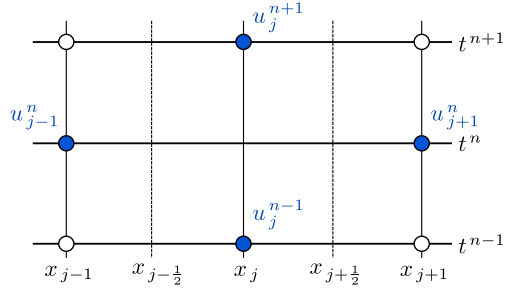}
\caption{Illustration of the \DF ~ space-time stencil.}
\label{fig:DF_stencil}
\end{figure}

\subsection{The lumped RC model}
\label{sec:RC_model}

The second model is the RC one. Interested readers are invited to consult \cite{Davies_2004,Fraisse_2002} for more details on this approach and \cite{Deconinck_2016,Reynders_2014,Jimenez_2009} for examples of recent applications in building physics.

\subsubsection{Formulation}

In the lumped RC approach, three temperature points are defined in the material as illustrated in Figure~\ref{fig:cell_RC}. Two temperatures are defined at the boundaries of the material:
\begin{align*}
T_{\,1} \,:\, \bigl[\, 0 \,,\, t_{\,\fin} \,\bigr] & \longrightarrow \, \mathbb{R}_{\,>\,0} \,, \\[4pt]
t & \longmapsto \, T(\,0\,,\,t\,) 
\end{align*} 
and
\begin{align*}
T_{\,3} \,:\, \bigl[\, 0 \,,\, t_{\,\fin} \,\bigr] & \longrightarrow \, \mathbb{R}_{\,>\,0} \,, \\[4pt]
t & \longmapsto \, T(\,1\,,\,t\,) \,.
\end{align*} 
The temperature $T_{\,2}$ is defined inside the cell $\mathcal{C}$ is $\ell \, \eqdef \, \displaystyle \frac{L}{2} \,$. According to the mean value theorem, this temperature is not necessarily in the middle of the wall.  
The formulation of the model is: 
\begin{align*}
\ell \cdot  c \cdot \od{T_{\,2}}{t} \egal k \cdot \Biggl(\, \pd{T}{x}\,\biggl|_{\,x \egal \frac{\,3 \cdot L}{4}} \moins \pd{T}{x}\,\biggl|_{\,x \egal \frac{L}{4}} \,\Biggr) \,.
\end{align*}
Using \textsc{Fourier}'s law to express the flux at the boundary of the cell, one obtains: 
\begin{align}
\label{eq:RC_model}
\ell^{\,2} \cdot  c \cdot \od{T_{\,2}}{t} \egal k \cdot \biggl(\, T_{\,3} \moins 2 \cdot T_{\,2} \plus T_{\,1} \,\biggr) \,.
\end{align}
It can be remarked that by integration, the partial differential heat diffusion equation is transformed into a simple ordinary equation in the RC lumped approach. Using a first order in space central discretisation for the boundary conditions given by equation~\eqref{eq:BC}, we obtain:
\begin{subequations}
\label{eq:RC_model_BC}
\begin{align}
\frac{k}{\ell} \cdot \Bigl(\, T_{\,2} \moins T_{\,1} \,\Bigr)&  \egal h_{\,L} \cdot \Bigl(\, T_{\,1} \moins T_{\,\infty\,,\,L} \,\Bigr)\,, \label{eq:RC_model_BC_L} \\[4pt]
\frac{k}{\ell} \cdot \Bigl(\, T_{\,3} \moins T_{\,2} \,\Bigr) & \egal - \, h_{\,R} \cdot \Bigl(\, T_{\,3} \moins T_{\,\infty\,,\,R} \,\Bigr) \,. \label{eq:RC_model_BC_R}
\end{align}
\end{subequations}
In the literature, this model is also referenced as R$2$C approach due to the straightforward electric analogy. 

\begin{figure}
\centering
\includegraphics[width=.45\textwidth]{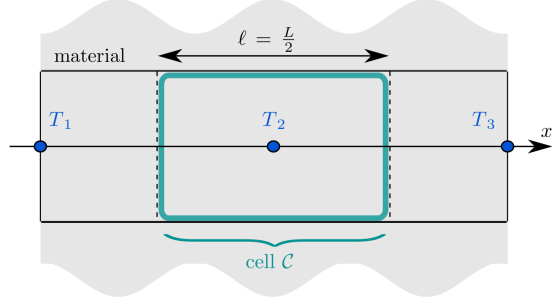}
\caption{Illustration of the three temperature defined in the lumped RC model.}
\label{fig:cell_RC}
\end{figure}

\subsubsection{Numerical Model}

The algorithm of the lumped RC model is developped in the \texttt{Matlab}\texttrademark ~environment. As performed in many works of literature \cite{Fraisse_2002,Kampf_2007,Naveros_2015}, the numerical model is not written in a dimensionless form. The physical dimensional variables are used. As for the previous numerical model, the discrete values of $T\,(\,t\,)$ are denoted using $T_{\,j}^{\,n} \, \eqdef \, T_{\,j}\,(\,t^{\,n}\,) \,, j \, \in \, \bigl\{\,1 \,,\, 2 \,,\, 3 \,\bigr\} \,$. The time discretisation parameter is designated by $\Delta t \,$. The ordinary differential Equation~\eqref{eq:RC_model} is approximated using an explicit \textsc{Euler} time scheme:
\begin{align*}
T_{\,2}^{\,n+1} \egal T_{\,2}^{\,n} \plus  \frac{k}{c \cdot \ell^{\,2}} \cdot \Delta t \cdot \biggl(\, T_{\,3}^{\,n}  \moins 2 \cdot T_{\,2 }^{\,n}  \plus T_{\,1}^{\,n} \,\biggr) \,.
\end{align*}
This choice of time discretisation scheme imposes a stability condition and the choice of the time step $\Delta t \,$:
\begin{align*}
\Delta t \, \leqslant \, \frac{1}{2} \cdot \frac{\ell^{\,2} \cdot c}{k} \,.
\end{align*}
To compute the temperatures $T_{\,1}$ and $T_{\,3} $ at the boundaries, the following equations are used:
\begin{align*}
\Bigl(\, h_{\,L} \plus \frac{k}{\ell}\,\Bigr) \cdot T_{\,1}^{\,n+1} & \moins \frac{k}{\ell} \cdot T_{\,2}^{\,n+1} \egal h_{\,L} \cdot T_{\,\infty\,,\, L} \,, \\[4pt]
\Bigl(\, h_{\,R} \plus \frac{k}{\ell}\,\Bigr) \cdot T_{\,3}^{\,n+1} & \moins \frac{k}{\ell} \cdot T_{\,2}^{\,n+1} \egal h_{\,R} \cdot T_{\,\infty\,,\, R} \,.
\end{align*}
In the end, the numerical model is written in a matrix form to compute the vector 
\begin{align*}
\bd{T}^{\,n+1} \, \eqdef \, \bigl[\, T_{\,1}^{\,n+1} \,,\, T_{\,2}^{\,n+1} \,,\, T_{\,3}^{\,n+1} \,\bigr]^{\,T} \,,
\end{align*}
using:
\begin{align}
\label{eq:RC_NM}
\bd{A} \cdot \bd{T}^{\,n+1} \egal \bd{B} \cdot  \bd{T}^{\,n} \plus \bd{Q} \,,
\end{align}
where
\begin{align*}
\bd{A} & \eqdef \, 
\begin{bmatrix}
\displaystyle h_{\,L} + \frac{k}{e} 
& \displaystyle - \frac{k}{e} & 0 \\[4pt]
0 & 1 & 0 \\[4pt] 
0 
& \displaystyle - \frac{k}{e} 
& \displaystyle h_{\,R} + \frac{k}{e}
\end{bmatrix} \,, 
&& \bd{B} \eqdef \, 
\begin{bmatrix}
0 & 0 & 0 \\[4pt]
\displaystyle \frac{k \cdot \Delta t}{c \cdot e^{\,2}}  
& \displaystyle 1 - 2 \cdot \frac{k \cdot \Delta t }{c \cdot e^{\,2}}
& \displaystyle \frac{k \cdot \Delta t}{c \cdot e^{\,2}}  \\[4pt] 
0 & 0 & 0
\end{bmatrix} \,, 
&& \bd{Q} \eqdef \, 
\begin{bmatrix}
h_{\,L} \cdot T_{\,\infty\,,\,L}^{\,n+1} \\
0 \\
h_{\,R} \cdot T_{\,\infty\,,\,R}^{\,n+1}
\end{bmatrix} \,.
\end{align*}
It can be noticed that the lumped RC model only requires the solution of three equations, while the complete model based on the heat diffusion equation needs $N_{\,x} \,$ calculations.

\section{Evaluating the reliability for the estimation of unknown parameter}
\label{sec:evaluating_reliability}

The procedure to evaluate the reliability of the mathematical model for the estimation of unknown parameters is illustrated in Figure~\ref{fig:procedure}. It is divided into two steps. The first one aims at generating experimental observations using a numerical model different from the DF or the RC ones. A total of $N_{\,s}$ samples of observations are generated \emph{in silico}. Then, for each sample, the parameter estimation problem is solved using the direct model based on the DF or the RC approaches. The suitable metrics to evaluate the reliability of each direct model for the estimation of unknown parameters are detailed in Section~\ref{sec:metrics}.

Before detailing the two steps, some preliminary definitions are provided. First, the singleton set $\Omega_{\,p}$ of the unique unknown parameter $p$ is defined by:
\begin{align*}
\Omega_{\,p} & \egal  \bigl\{\, p \,\bigr\} \,, \qquad p \, \in \, \mathbb{R}  \,.
\end{align*}
The distinction is done between the real parameter $p_{\,\real}$ used to generate the experimental observations. The identification problem aims to determine an estimate of parameter $ p_{\,\circ} \,$. If the model is reliable, it is expected that the difference between the real and estimated parameter to be as small as possible. An initial guess on the unknown parameter is required in the parameter estimation procedure, denoted $p_{\,\apr}\,$ using the subscript $_{\,\apr}$ for the \emph{a priori} estimation. 

To prove the theoretical identifiability of the unknown parameter $p$ the Structurally Globally Identifiable (SGI) property \cite{Walter_1982} is recalled. A parameter $p$ defined in the model $u$ is SGI if the following condition is satisfied:
\begin{align*}
\forall \, \bigl(\, p\,,\, p^{\,\prime} \,\bigr) \, \in \, \Omega_{\,p} \, \times \, \Omega_{\,p} \,, 
\qquad u\,\bigl(\,p\,\bigr) \equiv u\,\bigl(\,p^{\,\prime} \,\bigr)
\, \Longrightarrow \, 
p \equiv p^{\,\prime} \,.
\end{align*}
In other words, the mapping of $u$ is injective with respect to the parameter $p\,$.

We also define the ordered set of observation times:
\begin{align*}
\Omega_{\,t} & \egal \bigl(\, t_{\,1}^{\,\star} \,,\, \ldots \,,\, t_{\,k}^{\,\star} \,\bigr) \, \subset \, \bigl[\, 0 \,,\, t_{\,\fin}^{\,\star} \,\bigr]^{\,K}
 \,, && k \egal 1 \,,\,\ldots \,,\, K \,.
\end{align*}
From a practical point of view, the set $\Omega_{\,t}$ corresponds to the time grid where the experimental measurements are acquired. The point of coordinate $x_{\,\obs}^{\,\star} \,\in\,  \bigl[\,0\,,\,1\,\bigr]$ corresponds to the place where the sensor is located in the material to acquire the observation. The singleton set of sensor position is denoted by:
\begin{align*}
\Omega_{\,x} & \egal \bigl\{\, x_{\,obs}^{\,\star}  \,\bigr\} \,\subset \, \bigl[\,0\,,\,1 \,\bigr] \,, && x_{\,obs}^{\,\star}  \, \in\, \bigl[\,0\,,\,1\,\bigr] \,.
\end{align*}
In this work, only one sensor is settled so $\Omega_{\,x} \,\subset \, \mathbb{R}\,$.


\begin{figure}
\centering
\includegraphics[width=.65\textwidth]{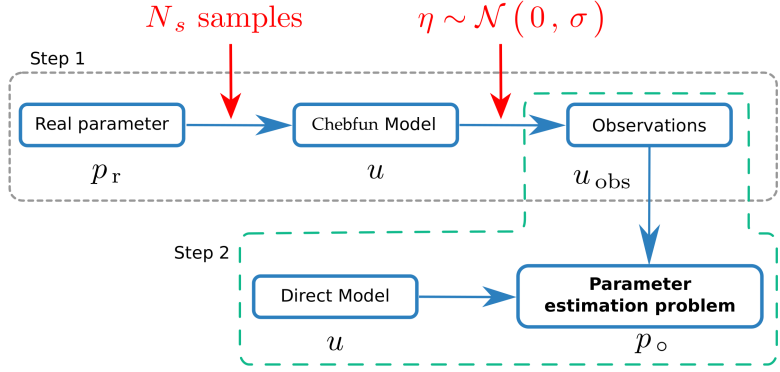}
\caption{Illustration of the procedure to evaluate the reliability of the mathematical model for the estimation of unknown parameters.}
\label{fig:procedure}
\end{figure}

\subsection{Step 1: generation of experimental observations}

The observations are generated with a numerical model based on pseudo--spectral approach using the \texttt{Matlab\texttrademark} open source toolbox \texttt{Chebfun} \cite{Chebfun_2014}. The model employs the function \texttt{pde23t} of \texttt{Chebfun} to compute a numerical reference solution $u_{\,\reff}$ of the partial derivative equation~\eqref{eq:heat_diffusion_dimless} based on the \textsc{Chebyshev} polynomials representation in space. The reference solution is computed using the real value of the parameter $p_{\,\real}\,$. It is directly obtained for the sensor location $x_{\,\obs}$ and the observation time set $\Omega_{\,t}\,$. Then, to obtain the $K$ observations $u_{\,\obs}\,$, a noise is added to simulate the experimental uncertainties due to the sensor design and location:
\begin{align*}
u_{\,\obs} \,:\, 
\Omega_{\,x} \,\times \, \Omega_{\,t} \,\times \, \Omega_{\,p} 
& \,\longrightarrow \, 
\Omega_{\,\obs} \,, \\[4pt] 
\Bigl(\, x_{\,\obs}^{\,\star}  \,,\, t_{\,k}^{\,\star}  \,,\, p_{\,\real} \,\Bigr) 
& \,\longmapsto \, 
u_{\,\reff} \, \bigl(\, x_{\,\obs}^{\,\star}  \,,\, t_{\,k}^{\,\star}   \,\bigr) \plus \eta \, \bigl(\, \omega \, \bigr) \,,
\end{align*}
where $\eta \,\sim\, \mathcal{N}\,\bigl(\, 0 \,,\, \sigma_{\,\obs} \,\bigr)$ is a noise following a \textsc{Gau}\ss ~normal distribution with  zero mean and standard deviation $\sigma_{\,\obs}\,$. The co-domain of $u_{\,\obs}$ verifies $\Omega_{\,\obs} \, \subset \, \mathbb{R}^{\,K} \,$.

\subsection{Step 2: solving the parameter estimation}

The parameter estimation problem is solved using the (numerically generated) experimental observations and the solution of the direct model $u_{\,\dir}\,$. The latter is defined by:
\begin{align*}
u_{\,\dir} \,:\, 
\Omega_{\,x} \,\times \, \Omega_{\,t} \,\times \, \Omega_{\,p} 
& \,\longrightarrow \, 
\Omega_{\,\dir} \,, \\[4pt] 
\Bigl(\, x_{\,\obs}^{\,\star}  \,,\, t_{\,k}^{\,\star}  \,,\, p \,\Bigr) 
& \,\longmapsto \, 
u_{\,\dir}\, \Bigl(\, x_{\,\obs}^{\,\star}  \,,\, t_{\,k}^{\,\star}  \,,\, p \,\Bigr)  \,.
\end{align*}
It is computed using the DF model (described in Section~\ref{sec:DF_model}) or the RC one (described in Section~\ref{sec:RC_model}). The domain and the co-domain of $u_{\,\dir}$ verifies $\mathrm{dom} \ u_{\,\obs} \, \equiv \, \mathrm{dom} \ u_{\,\dir}$ and $\Omega_{\,\dir} \, \subset \, \mathbb{R}^{\,K} \,$, respectively.
The identification problem aims at computing the estimated parameter $ p_{\,\circ} \,$ verifying:
\begin{align}
\label{eq:estimated_parameter_def}
p_{\,\circ} \, \eqdef \, \arg \; \min_{\,\Omega_{\,p}} \;  J \,,
\end{align}
where $J$ is the so-called cost function defined by the least square estimator:
\begin{align}
\label{eq:cost_function}
J \;:\, \Omega_{\,\dir} \, \times \, \Omega_{\,\obs} \,&   \longrightarrow \, \mathbb{R}_{\,\geqslant \,0} \,,\\
 \bigl(\, u_{\,\dir} \,,\, u_{\,\obs} \,\bigr) & \,\longmapsto \, \Bigl| \Bigl|\, u_{\,\dir}  \moins u_{\,\obs} \,\Bigr| \Bigr|_{\,2} \,, \nonumber
\end{align}
where $\Bigl| \Bigl|\, \bullet \,\Bigr| \Bigr|_{\,2}$ is the least square estimator $\mathcal{L}_{\,2}$ defined by:
\begin{align*}
\Bigl| \Bigl|\, \bullet \,\Bigr| \Bigr|_{\,2} \,:\; & y 
\, \longmapsto \, 
\frac{1}{t_{\,\fin}} \cdot \int_{\,\Omega_{\,t}} \; \Bigl(\, y\,\bigl(\,t \,\bigr) \,\Bigr)^2 \, \mathrm{d}t\,.
\end{align*}
The dependency of the cost function $J$ on the unknown parameter $p$ can be understood by the diagram of mapping illustrated in Figure~\ref{fig:mapping_J}. The minimization of the cost function~\eqref{eq:estimated_parameter_def} is realized using the \textsc{Gau}\ss ~algorithm \cite{Beck_1977,Ozisik_2000,Kabanikhin_2008b}. Specifically, the necessary condition for the minimum of $J$ is:
\begin{align*}
\pd{J}{p} \egal \textbf{0} \,,
\end{align*}
which is equivalent to
\begin{align}
\label{eq:diffJ_diffp}
\frac{1}{t_{\,\fin}} \cdot  \int_{\,\Omega_{\,t}} \; 2 \cdot \pd{u_{\,\dir}}{p} \cdot \Bigl(\, u_{\,\dir} \,(\,p\,)  \moins u_{\,\obs} \, \Bigr) \, \mathrm{d}t  \egal \textbf{0} \,.
\end{align}
Assuming we have a candidate for the unknown parameter $p_{\,m} \,$, the \textsc{Taylor} expansion gives:
\begin{align*}
u_{\,\dir} \,(\,p\,) \egal u_{\,\dir}\,(\,p_{\,m}\,) \plus \pd{u}{p}\,\biggl|_{\,p \egal p_{\,m}} \cdot \, \Bigl(\, p \moins p_{\,m} \,\Bigr) \plus \mathcal{O}\,\biggl(\, \Bigl(\,p \moins p_{\,m} \,\Bigr)^{\,2} \,\biggr) \,.
\end{align*}
So, Equation~\eqref{eq:diffJ_diffp} after truncation becomes:
\begin{align}
\label{eq:Gauss_linearization}
\int_{\,\Omega_{\,t}} \; 2 \cdot \pd{u_{\,\dir}}{p} \cdot 
\Biggl(\, u_{\,\dir}\,(\,p_{\,m}\,) \plus \pd{u_{\,\dir}}{p}\,\biggl|_{\,p \egal p_{\,m}} \cdot \, \Bigl(\, p \moins p_{\,m} \,\Bigr)
\moins u_{\,\obs} \, \Biggr) \, \mathrm{d}t  \egal \textbf{0} \,.
\end{align}
Equation~\eqref{eq:Gauss_linearization} provides the \textsc{Gau}\ss ~linearization to compute a candidate $p$ better than $p_{\,m}$ to minimize the cost function $J\,$. To indicate the iterative procedure, the notation is slightly changed and parameter $p \, \leftarrow \, p_{\,m+1} \,$. Thus, the candidate $p_{\,m+1}$ is computed by forcing equation~\eqref{eq:Gauss_linearization} to vanish:
\begin{align*}
p_{\,m+1} \egal p_{\,m} \plus \frac{u_{\,\obs} \moins u_{\,\dir}}{\displaystyle \pd{u_{\,\dir}}{p}\,\biggl|_{\,p \egal p_{\,m}}} \,.
\end{align*}
The computation of the candidate $p_{\,m+1}$ requires the knowledge of the sensitivity function $\displaystyle \pd{u_{\,\dir}}{p} \,$. For this, for each direct model (DF or RC ones), the sensitivity equation is obtained by differentiating the main equation with respect to the unknown parameter $p \,$. All sensitivity equations for the two direct models and some comments on their resolution are provided in Appendix~\ref{app:sensitivity_equations}. The iterative procedure is implemented starting from the initial guess $p_{\,\apr}\,$. Two stopping criteria $\gamma_{\,1}$ and $\gamma_{\,2}$ are defined on the magnitude of changes of the cost function and unknown parameter:
\begin{align*}
\gamma_{\,1} \,\Bigl(\,p_{\,m} \,,\, p_{\,m+1} \,\Bigr) & \eqdef \, \frac{\Bigl|\Bigl|\, p_{\,m+1} \moins p_{\,m} \, \Bigr|\Bigr|_{\,2} }{\Bigl|\Bigl|\, p_{\,m} \, \Bigr|\Bigr|_{\,2}} \,, \\[4pt]
\gamma_{\,2} \,\Bigl(\,p_{\,m} \,,\, p_{\,m+1} \,\Bigr) & \eqdef \, 
\frac{ \Bigl|\Bigl|\, u_{\,\dir} \,(\,p_{\,m+1}\,)  \moins u_{\,\obs} \, \Bigr|\Bigr|_{\,2} 
\moins \Bigl|\Bigl|\, u_{\,\dir} \,(\,p_{\,m}\,)  \moins u_{\,\obs} \, \Bigr|\Bigr|_{\,2} }{ \Bigl|\Bigl|\, u_{\,\dir} \,(\,p_{\,m}\,)  \moins u_{\,\obs} \, \Bigr|\Bigr|_{\,2} } \,.
\end{align*}
The algorithm stop when the following conditions are reached:
\begin{align*}
\Bigl(\, \gamma_{\,1} & \leqslant \, \eta_{\,1} \,\Bigr) 
&& \& 
&& \Bigl(\, \gamma_{\,2} \leqslant \, \eta_{\,2} \,\Bigr)  \,.
\end{align*}
where $\eta_{\,1}$ and $\eta_{\,2}$ are small positive values set in this work to $\eta_{\,1} \egal \eta_{\,2} \egal 10^{\,-6} \,$. 

\begin{figure}
\centering
\includegraphics[width=.65\textwidth]{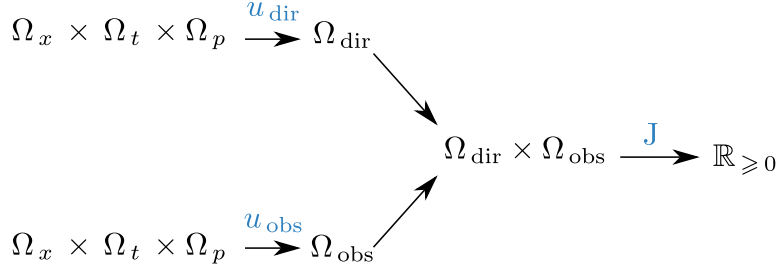}
\caption{Diagram of mapping involving the cost function $J \,$.}
\label{fig:mapping_J}
\end{figure}

\subsection{Metrics for the evaluation of the reliability}
\label{sec:metrics}

Several metrics are defined for the evaluation of the reliability of the two models to estimate one unknown parameter among $\bigl\{\, c \,,\, k \,,\, h_{\,L} \,\bigr\}\,$. Since the estimation of the unknown parameter is realized for $N_{\,s}$ samples of observations, it is possible to compute classical statistical metrics. The expectation $\mathsf{E}\,\bigl[\, - \, \bigr]$ and the standard deviation  $\sigma\,\bigl[\, - \, \bigr]$ of the random variable $y$ are defined by:
\begin{align*}
\mathsf{E} \, \bigl(\,y \,\bigr) & \eqdef \, \frac{1}{N_{\,s}} \, \sum_{s\egal 1}^{N_{\,s}} \ y_{\,s} \,, 
&& \sigma^{\,2} \, \bigl(\,y \,\bigr) \, \eqdef \, \mathsf{E} \,\biggl(\, \Bigl(\,y \moins \mathsf{E} \, \bigl(\, y \,\bigr) \,\Bigr)^{\,2} \,\biggr) \,. 
\end{align*}
These metrics can be applied to the ratio $\displaystyle \frac{p_{\,\circ}}{p_{\,\real}}$ between the estimated and real parameters and to the number of iterations $N_{\,m}$ or the computational (CPU) time $t_{\,\mathrm{CPU}}$ required for the algorithm to estimate the parameter. The latter is measured using \texttt{Matlab\texttrademark}  environment on a computer equipped with Intel i$7-6820$HQ CPU, $2.70$GHz and $32\,$GB of RAM.

\section{Case studies}
\label{sec:case_studies}

The reliability of the mathematical model is evaluated to estimate one unknown parameter among $\bigl\{\, c \,,\, k \,,\, h_{\,L} \,\bigr\}\,$, two others being fixed. Five types of usual building materials are considered as summarized in Table~\ref{tab:material}. The thickness of the material is $L \egal 22 \ \mathsf{cm}\,$. The initial condition is $T_{\,0} \egal 20 \ \mathsf{^{\,\circ}C}\,$. At both boundary conditions, the ambient temperatures follow the sinusoidal variations:
\begin{align*}
T_{\,\infty\,,\,L} \, \bigl(\,t\,\bigr) 
& \egal T_{\,0} \plus 10 \cdot \sin \biggl(\, \frac{2 \cdot \pi}{20 \cdot 3600} \cdot t \,\biggr) 
\plus 10  \cdot \sin \biggl(\, \frac{2 \cdot \pi}{2 \cdot 3600} \cdot t \,\biggr) \,, \\[4pt]
T_{\,\infty\,,\,R} \, \bigl(\,t\,\bigr) 
& \egal T_{\,0} \plus 20 \cdot \tanh \biggl(\, \frac{1}{4 \cdot 3600} \cdot t \,\biggr) 
\moins 10  \cdot \sin \biggl(\, \frac{2 \cdot \pi}{4 \cdot 3600} \cdot t \,\biggr) \,,
\end{align*}
which are illustrated in Figure~\ref{fig:BC_ft}. The heat transfer coefficient at the right boundary equals to $h_{\,R} \egal 5 \ \mathsf{W\,/\,(m^{\,2} \cdot K)}\,$. 

For each case, $N_{\,s} \egal 10^{\,4} $ samples of observations are generated with a noise of standard deviation $\sigma_{\,\obs} \egal 0.2 \ \mathsf{^{\,\circ}C}\,$, corresponding to usual uncertainty of temperature measurement. The point of observation is the middle of the wall $x_{\,\obs} \egal 11 \ \mathsf{cm}\,$. The time grid of each sample of observations is set as $t_{\,k} \egal k \cdot 360 \ \mathsf{s}\,, k \, \in \, \bigl\{\,0 \,,\, \ldots \,,\, 200 \,\bigr\} \, \subset \, \mathbb{N}_{\,0} \,$. Thus, each sample includes $K \egal 201$ observations obtained with a time step of $360 \ \mathsf{s}\,$.

The discretisation parameter are set to $\Delta t \egal 3.6 \ \mathsf{s}$ and $\Delta x \egal 2.2 \ \mathsf{mm}$ for the \DF ~model. For the RC lumped model, the time discretisation parameter is also $\Delta t \egal 3.6 \ \mathsf{s}\,$.

\begin{figure}
\centering
\includegraphics[width=.45\textwidth]{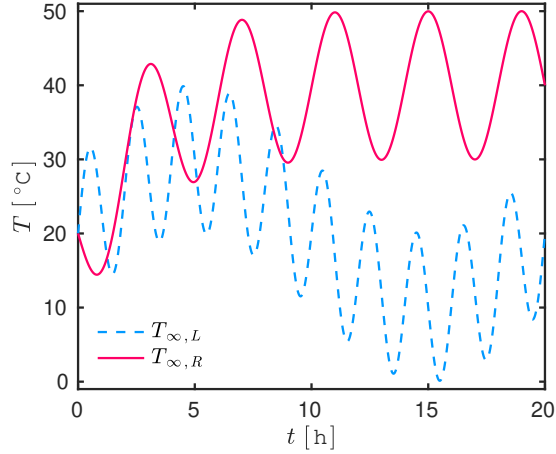}
\caption{Time variation of the boundary conditions.}
\label{fig:BC_ft}
\end{figure}

\begin{table}
\centering
\caption{Real properties of the material considered for the case studies.}
\label{tab:material}
\setlength{\extrarowheight}{.5em}
\begin{tabular}[l]{@{} c c c c}
\hline
\hline
\textit{Identification}
& \textit{Volumetric heat capacity }
& \textit{Thermal conductivity} 
& \textit{Material} \\
\textit{$N^{\,\circ}$}
& \textit{$c \ \unitt{MJ}{kg \cdot K}$}
& \textit{$k \ \unitt{W}{m \cdot K}$} 
& \textit{type} \\
1 
& $5 \e{-2}$
& $5 \e{-2}$
& \textbf{insulation} \\
2 
& $5 \e{-1}$
& $5 \e{-1}$
& \textbf{wood} \\
3 
& $1.5$
& $1$
& \textbf{brick} \\
4 
& $2.0$
& $1.5$
& \textbf{concrete} \\
5 
& $2.5$
& $2.5$
& \textbf{stone} \\
\hline
\hline
\end{tabular}
\end{table}

\subsection{Estimation of the volumetric heat capacity}

The purpose is to estimate the thermal capacity $c\,$ for the five types of materials. 
Before generating the experimental observations and performing the estimation, it is necessary to prove the identifiability of the parameter $c$ for both models using the SGI property. First, the demonstration is realized for the DF model. \\~
\begin{proposition}
The parameter $c$ is identifiable in Equation~\eqref{eq:heat1d}.
\end{proposition}
\begin{proof}
We assume an observable $T\,\bigl(\,-\,,\,-\,\bigr)$ verifies the model:
\begin{align}
\label{eq:SGI_heat_eq_c}
c \cdot \pd{T}{t} \egal k \cdot \pd{^{\,2} T}{x^{\,2}} \,.
\end{align}
Another observable, denoted with a superscript $^{\,\prime}\,$, obtained with another parameter $c^{\,\prime}$ holds:
\begin{align}
\label{eq:SGI_heat_eq_c_prime}
c^{\,\prime} \cdot \pd{T^{\,\prime}}{t} \egal k \cdot \pd{^{\,2} T^{\,\prime}}{x^{\,2}} \,.
\end{align}
If $T \, \equiv \, T^{\,\prime}$ then $\displaystyle \pd{T}{t} \, \equiv \, \pd{T^{\,\prime}}{t}$ and $\displaystyle  \pd{^{\,2} T}{x^{\,2}} \, \equiv \, \pd{^{\,2} T^{\,\prime}}{x^{\,2}} \,$. Thus, from equations~\eqref{eq:SGI_heat_eq_c} and \eqref{eq:SGI_heat_eq_c_prime}, we obtain:
\begin{align*}
\Bigl(\, c \moins c^{\,\prime} \,\Bigr) \cdot \pd{T}{t} \egal 0
\end{align*}
so $c \, \equiv \, c^{\,\prime}$ and parameter $c$ is SGI. 
\end{proof}
Now, the identifiability is proven for the RC model. \\~
\begin{proposition}
The parameter $c$ is identifiable in Equation~\eqref{eq:RC_model}.
\end{proposition}
\begin{proof}
We assume an observable $T\,\bigl(\,-\,\bigr)$ obtained from the RC model:
\begin{align}
\label{eq:SGI_RC_c_model}
e^{\,2} \cdot  c \, \od{T_{\,2}}{t} \egal k \cdot \biggl(\, T_{\,3} \moins 2 \cdot T_{\,2} \plus T_{\,1} \,\biggr) \,.
\end{align}
Another observable, denoted with a superscript $^{\,\prime}\,$, obtained with another parameter $c^{\,\prime}\,$:
\begin{align}
\label{eq:SGI_RC_c_model_prime}
e^{\,2} \cdot  c^{\,\prime} \, \od{T_{\,2}^{\,\prime}}{t} \egal k \cdot \biggl(\, T_{\,3}^{\,\prime} \moins 2 \cdot T_{\,2}^{\,\prime} \plus T_{\,1}^{\,\prime} \,\biggr) \,.
\end{align}
If $T \, \equiv \, T^{\,\prime}$ then $\displaystyle \od{T}{t} \, \equiv \, \od{T^{\,\prime}}{t}\,$. Thus, from Equations~\eqref{eq:SGI_RC_c_model} and \eqref{eq:SGI_RC_c_model_prime}, one obtains:
\begin{align*}
\Bigl(\, c \moins c^{\,\prime} \,\Bigr) \cdot \od{T_{\,2}}{t} \egal 0
\end{align*}
so $c \, \equiv \, c^{\,\prime}$ and parameter $c$ is SGI in the RC model. 
\end{proof}

The experimental observations are generated using the real heat capacity $c_{\,\real}$ given in Table~\ref{tab:material}. The thermal conductivity is a known parameter given for each material in the same Table. The heat transfer coefficient at the left boundary conditions is set to $h_{\,L} \egal 15 \ \mathsf{W\,/\,(m^{\,2} \cdot K)}\,$. For the solution of the parameter estimation problem, the initial guess of $c$ is fixed in the algorithm as $c_{\,\apr} \egal 0.1 \cdot c_{\,\real} \,$.

The expectation of the estimated parameter $c_{\,\circ}$ using both mathematical models DF and RC is compared with the real parameter $c_{\,\real}$ in Figure~\ref{fig:c_P_frhoc}. More detailed results are provided in Table~\ref{tab:c_results_stats}. The DF model allows to estimate accurately the unknown parameter $c$. For the five materials, the expectation of the ratio between the estimated and real parameter approximately equal to $1\,$. For the RC model, the estimation lacks of accuracy for all materials. There is slight decrease of the expectation of the estimated parameter with the increase of volumetric heat capacity $c_{\,\real}\,$. For the material $N^{\,\circ}\,5\,$, the RC lumped model estimates a parameter with almost $50\%$ of the relative error. As reported in Table~\ref{tab:c_results_stats}, for both models the standard deviation of the estimated parameter is small. 
The number of iterations required for the estimation of the parameter are illustrated in Figure~\ref{fig:c_Nk_frhoc} with more detailed results in Table~\ref{tab:c_results_stats}. Mainly, the DF model requires fewer iterations to estimate the parameter than the RC one. The number of iterations is eight times more in average for the RC model, while it seems to decrease for the DF model together with the heat capacity. Figure~\ref{fig:c_cpu_frhoc} gives the computational time needed by the algorithm to converge to the estimated parameter. The DF model has a higher computational cost, around $2.5 \ \mathsf{s}$ for one estimation. Even if the algorithm based on the DF approach needs fewer iterations, these computational effort differences are due to the construction of each numerical model. Indeed, at each time iteration, $N_{\,x} \egal 100$ equations are computed for DF approach while only $3$ for the RC model. It should be recalled that the same time discretisation parameters were used for both models. 

An insight of the results for the materials $1$ and $3$ is illustrated in Figures~\ref{fig:c_T_all_ft_mat1} and \ref{fig:c_T_all_ft_mat3}. The time evolution of the temperature expectation computed with the estimated parameter for both models is compared to the experimental observations. The RC model lacks of accuracy to represent the physical phenomena. On the other hand, there is a satisfactory agreement between the predictions of the DF model and the experimental observations. Figures~\ref{fig:c_crit1_fk} and \ref{fig:c_crit2_fk} show the convergence of the algorithm relatively to the number of iterations for both models. It can be remarked that the DF model convergence is faster than for the RC one. These results may be due to lack of accuracy in the computation of the sensitivity equations by the RC model.

\begin{figure}
\centering
\includegraphics[width=.45\textwidth]{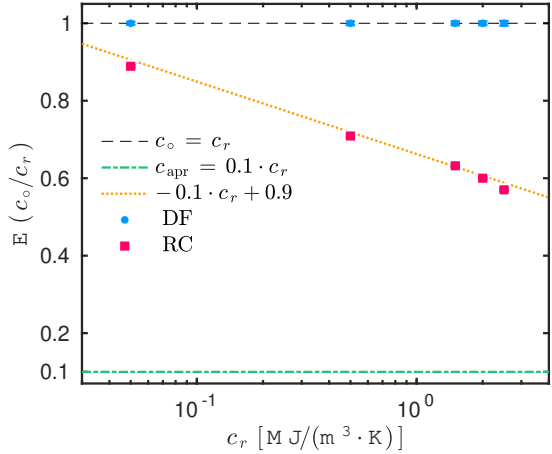}
\caption{Comparison of the expectation of the estimated parameter $c_{\,\circ}$ with the real parameter $c_{\,\real}\,$.}
\label{fig:c_P_frhoc}
\end{figure}

\begin{figure}
\centering
\subfigure[\label{fig:c_Nk_frhoc}]{\includegraphics[width=.45\textwidth]{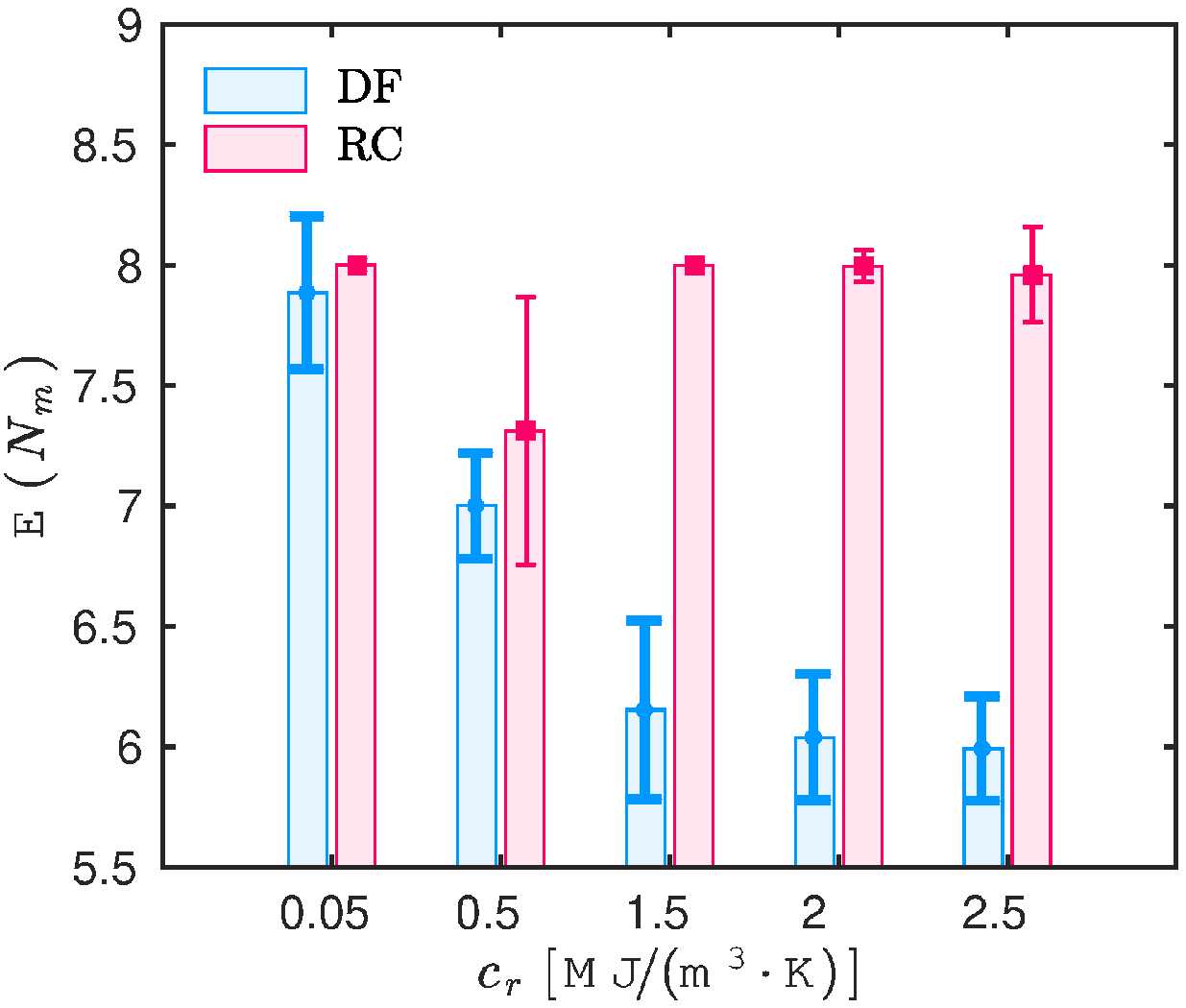}}  \hspace{0.2cm}
\subfigure[\label{fig:c_cpu_frhoc}]{\includegraphics[width=.45\textwidth]{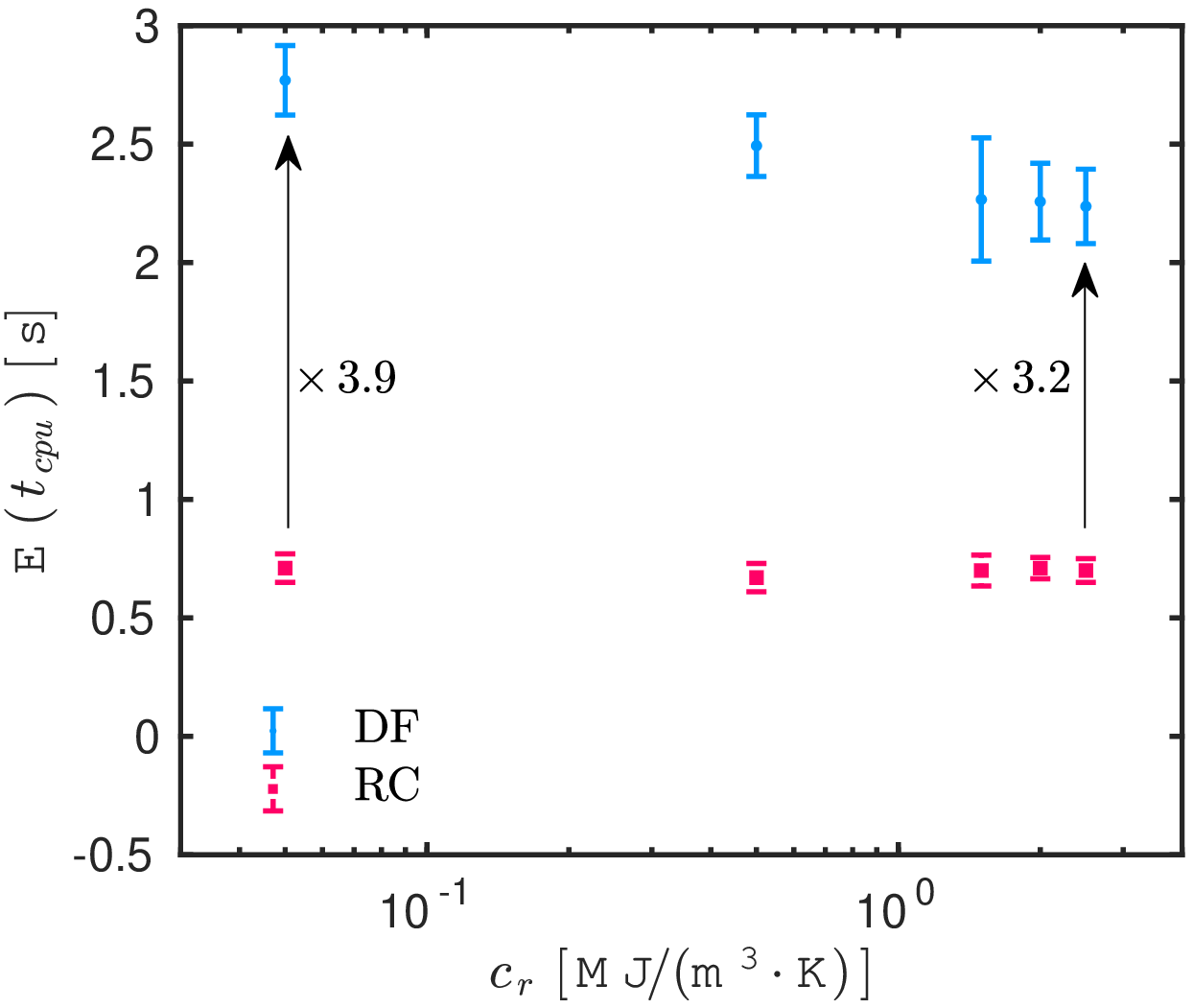}}
\caption{Variation of the expectation of the number of iteration $N_{\,m}$ \emph{(a)} and the computational time $t_{\,\mathrm{cpu}}$ \emph{(b)} for the algorithm to estimate the unknown parameter $c$ for the five types of materials.}
\end{figure}

\begin{figure}
\centering
\subfigure[\label{fig:c_T_all_ft_mat1}]{\includegraphics[width=.45\textwidth]{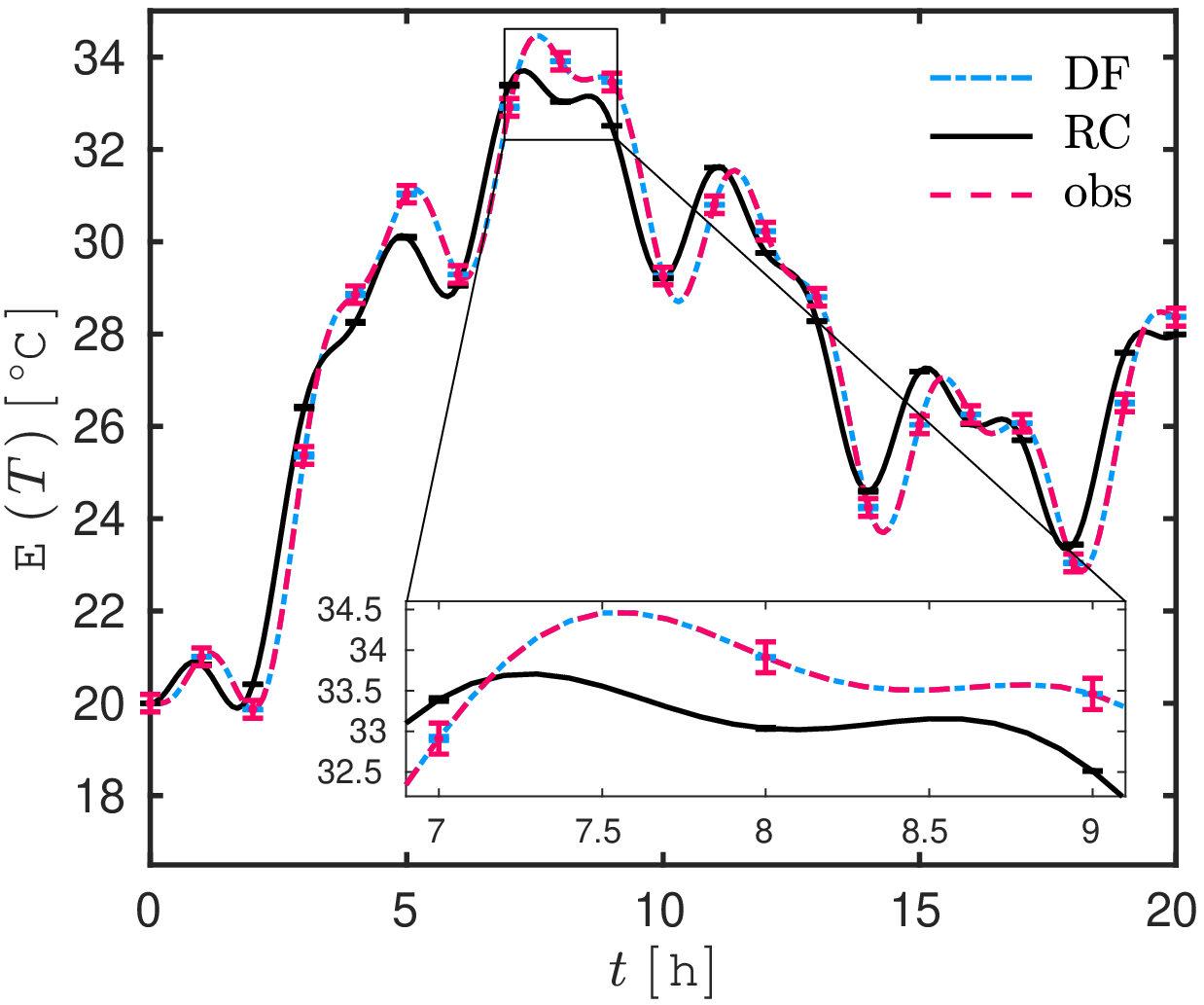}}  \hspace{0.2cm}
\subfigure[\label{fig:c_T_all_ft_mat3}]{\includegraphics[width=.45\textwidth]{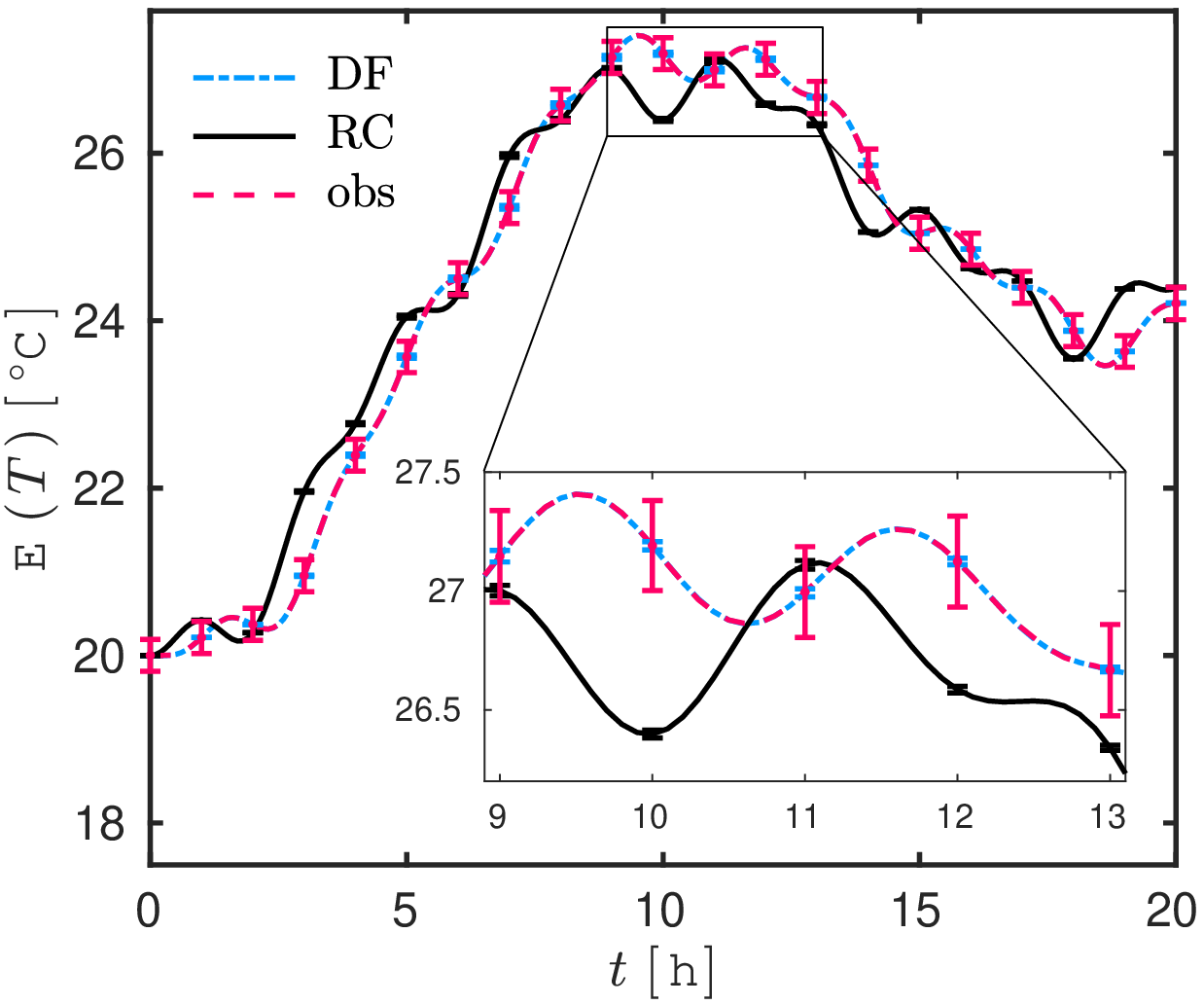}}
\caption{Time evolution of the temperature at $x \egal x_{\,\obs} \egal 11 \ \mathsf{cm}$ for material $1$ \emph{(a)} and material $3$ \emph{(b)} computed with the numerical model for $c \egal c_{\,\circ}\,$.}
\end{figure}

\begin{figure}
\centering
\subfigure[\label{fig:c_crit1_fk}]{\includegraphics[width=.45\textwidth]{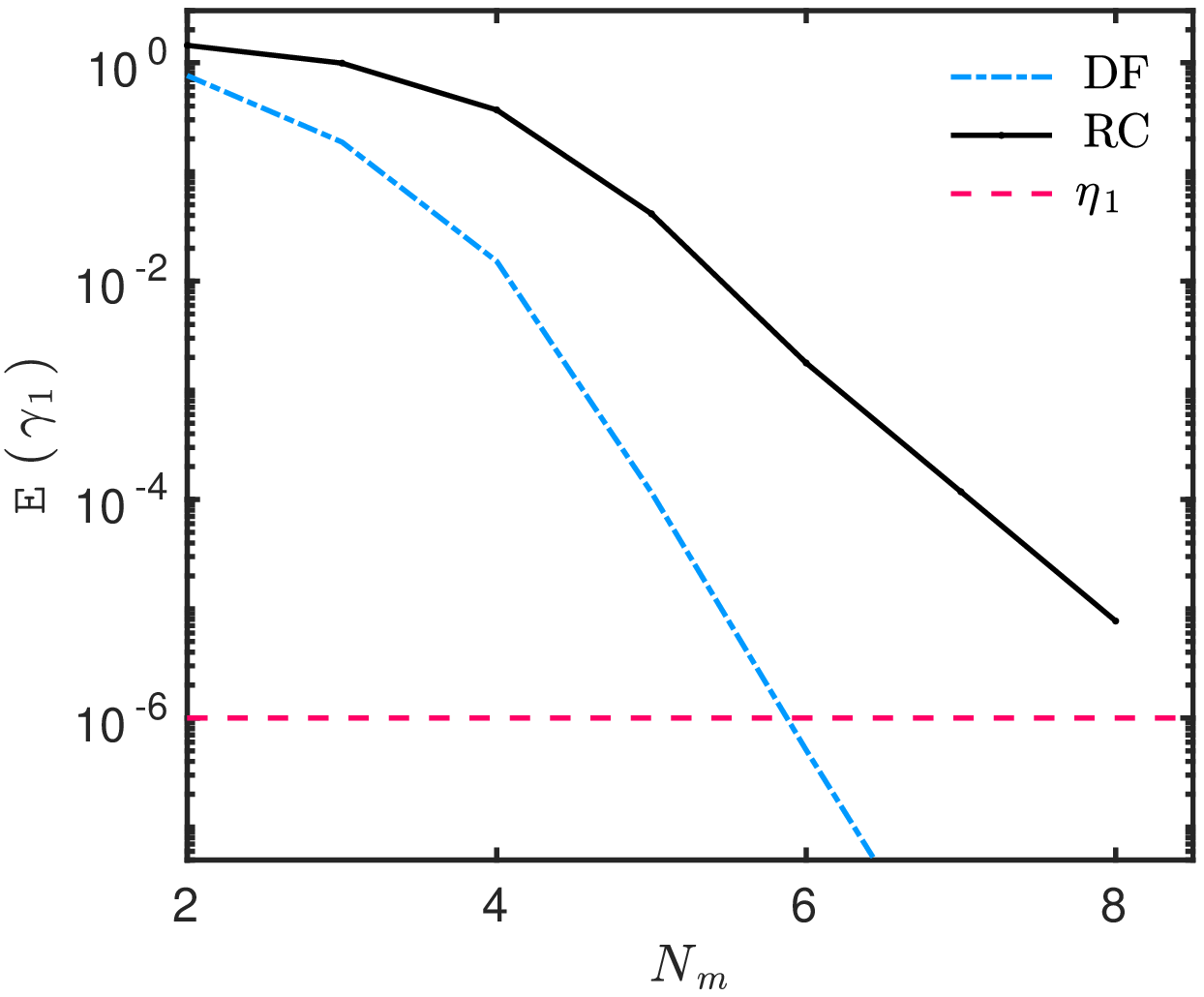}}  \hspace{0.2cm}
\subfigure[\label{fig:c_crit2_fk}]{\includegraphics[width=.45\textwidth]{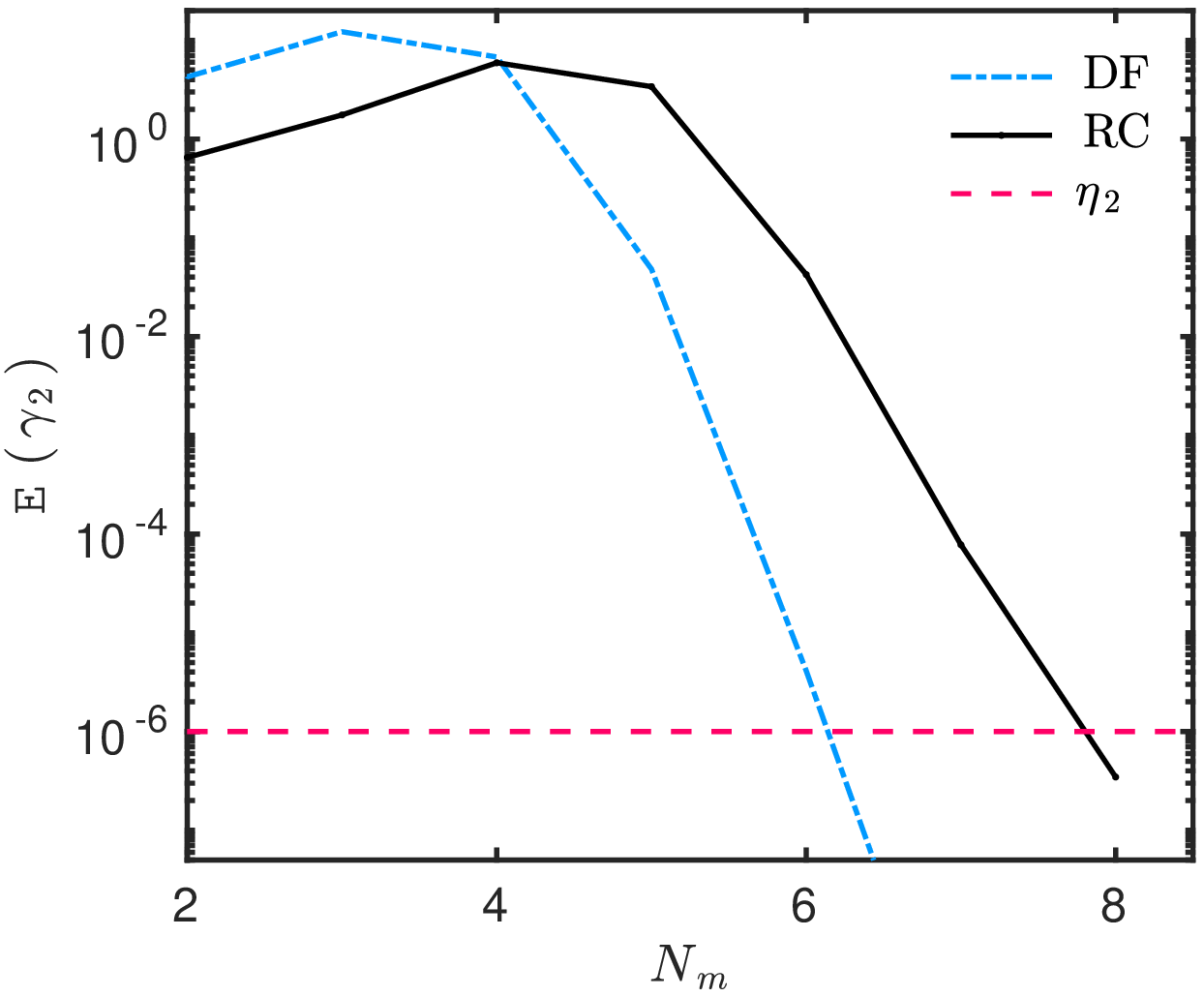}}
\caption{Variation of the expectation of the convergence criteria $\gamma_{\,1}$ \emph{(a)} and $\gamma_{\,2}$ \emph{(b)} for the estimation of the unknown parameter $c$ for the material $3$.}
\end{figure}

\begin{table}
\centering
\caption{Results for the estimation of the unknown volumetric heat capacity $c \,$.}
\label{tab:c_results_stats}
\setlength{\extrarowheight}{.5em}
\begin{tabular}[l]{@{} c|cc|cc|cc|cc|cc|cc}
\hline
\hline
& \multicolumn{4}{c|}{\textit{Ratio $\displaystyle \frac{c_{\,\circ}}{c_{\,\real}}$}}
& \multicolumn{4}{c|}{\textit{Number of iterations $N_{\,m}$}}
& \multicolumn{4}{c}{\textit{Computational time $t_{\,\mathrm{CPU}} \ \unit{s}$}} \\[8pt]
\textit{Material}
& \multicolumn{2}{c|}{\textit{DF model}}
& \multicolumn{2}{c|}{\textit{RC model}}
& \multicolumn{2}{c|}{\textit{DF model}}
& \multicolumn{2}{c|}{\textit{RC model}}
& \multicolumn{2}{c|}{\textit{DF model}}
& \multicolumn{2}{c}{\textit{RC model}} \\
\textit{Identification}
& $\mathsf{E}$
& $\sigma$
& $\mathsf{E}$
& $\sigma$
& $\mathsf{E}$
& $\sigma$
& $\mathsf{E}$
& $\sigma$
& $\mathsf{E}$
& $\sigma$
& $\mathsf{E}$
& $\sigma$ \\
\hline
1
& $1.0$ & $0.004$
& $0.89$ & $0.004$
& $7.9$ & $0.32$
& $8$ & $0.03$ 
& $2.8$ & $0.15$
& $0.7$ & $0.05$ \\
2
& $1.0$ & $0.005$
& $0.71$ & $0.003$
& $7.0$ & $0.22$
& $7.3$ & $0.56$ 
& $2.5$ & $0.13$
& $0.7$ & $0.05$ \\
3
& $1.0$ & $0.005$
& $0.63$ & $0.003$
& $6.2$ & $0.38$
& $8$ & $0$ 
& $2.3$ & $0.26$
& $0.7$ & $0.05$ \\
4
& $1.0$ & $0.005$
& $0.6$ & $0.003$
& $6.0$ & $0.26$
& $8$ & $0.06$ 
& $2.3$ & $0.16$
& $0.7$ & $0.05$ \\
5
& $1.0$ & $0.006$
& $0.57$ & $0.003$
& $6.0$ & $0.22$
& $8$ & $0.2$ 
& $2.4$ & $0.16$
& $0.7$ & $0.05$ \\
\hline
\hline
\end{tabular}
\end{table}

\subsection{Estimation of the thermal conductivity}

The issue is now to estimate the thermal conductivity $k\,$ for the five material. Let us prove the identifiability of the parameter in each model. The demonstration is similar to the one for the previous case study. 
\begin{proposition}
For the DF model, the parameter $k$ is identifiable in equation~\eqref{eq:heat1d}.
\end{proposition}
\begin{proof}
We assume an observable $T\,\bigl(\,-\,,\,-\,\bigr)$ verifies the model:
\begin{align}
\label{eq:SGI_heat_eq_k}
c \cdot \pd{T}{t} \egal k \cdot \pd{^{\,2} T}{x^{\,2}} \,.
\end{align}
Another observable, denoted with a superscript $^{\,\prime}\,$, obtained with another parameter $k^{\,\prime}$ is detained:
\begin{align}
\label{eq:SGI_heat_eq_k_prime}
c \cdot \pd{T^{\,\prime}}{t} \egal k^{\,\prime} \cdot \pd{^{\,2} T^{\,\prime}}{x^{\,2}} \,.
\end{align}
If $T \, \equiv \, T^{\,\prime}$ then $\displaystyle \pd{T}{t} \, \equiv \, \pd{T^{\,\prime}}{t}$ and $\displaystyle  \pd{^{\,2} T}{x^{\,2}} \, \equiv \, \pd{^{\,2} T^{\,\prime}}{x^{\,2}} \,$. Thus, from equations~\eqref{eq:SGI_heat_eq_k} and \eqref{eq:SGI_heat_eq_k_prime}, we obtain:
\begin{align*}
\Bigl(\, k \moins k^{\,\prime} \,\Bigr) \cdot  \pd{^{\,2} T}{x^{\,2}} \egal 0 \,.
\end{align*}
Thus, $k \, \equiv \, k^{\,\prime}$ and parameter $k$ is SGI. 
\end{proof}
Secondly, the identifiability is proven for the RC model. \\~
\begin{proposition}
The parameter $k$ is identifiable in equation~\eqref{eq:RC_model}.
\end{proposition}
\begin{proof}
We assume an observable $T$ obtained for the RC model:
\begin{align}
\label{eq:SGI_RC_model_k}
e^{\,2} \cdot  c \, \od{T_{\,2}}{t} \egal k \cdot \biggl(\, T_{\,3} \moins 2 \cdot T_{\,2} \plus T_{\,1} \,\biggr) \,.
\end{align}
Another observable, denoted with a superscript $^{\,\prime}\,$, obtained with another parameter $k^{\,\prime}$ holds:
\begin{align}
\label{eq:SGI_RC_model_k_prime}
e^{\,2} \cdot  c \, \od{T_{\,2}^{\,\prime}}{t} \egal k^{\,\prime} \cdot \biggl(\, T_{\,3}^{\,\prime} \moins 2 \cdot T_{\,2}^{\,\prime} \plus T_{\,1}^{\,\prime} \,\biggr) \,.
\end{align}
If $T \, \equiv \, T^{\,\prime}$ then $\displaystyle \od{T}{t} \, \equiv \, \od{T^{\,\prime}}{t}\,$. Thus, from equations~\eqref{eq:SGI_RC_model_k} and \eqref{eq:SGI_RC_model_k_prime}, one obtain:
\begin{align*}
\Bigl(\, k \moins k^{\,\prime} \,\Bigr) \cdot \biggl(\, T_{\,3} \moins 2 \cdot T_{\,2} \plus T_{\,1} \,\biggr) \egal 0 \,.
\end{align*}
Since $\biggl(\, T_{\,3} \moins 2 \cdot T_{\,2} \plus T_{\,1} \,\biggr) \, \neq \, 0\,$, one can deduce that $k \, \equiv \, k^{\,\prime}$ and that parameter $k$ is SGI in the RC model. 
\end{proof}

Before generating the experimental observations, an important remark can be formulated. From an mathematical point of view, it can be noted that only the ratio $\frac{k}{c}$ is identifiable in each model. One could question the necessity of evaluating the reliability of the models for the estimation of $k$ since the results might be similar to the ones obtained for the parameter $c\,$. Nevertheless, from a practical point of view, once estimated, these parameters are used in computational tools for evaluating the building energy requirements in the context of thermal regulations. Thus, it is of major importance to see evaluate the accuracy of each model to recover each parameters.

With this results, the experimental observations can be generated using the real thermal conductivity $k_{\,\real}$ given in Table~\ref{tab:material}. For each material, $N_{\,s}$  experimental observations are produced. The heat capacity is a given parameter from Table~\ref{tab:material} for each case. The heat transfer coefficient at the left boundary is also fixed to $h_{\,L} \egal 15 \ \mathsf{W\,/\,(m^{\,2} \cdot K)}\,$. In the algorithm to estimate the unknown parameter, the initial guess is prescribed as $k_{\,\apr} \egal 0.1 \cdot k_{\,\real} \,$.

Figure~\ref{fig:k_P_flamb} compares the expectation over the $N_{\,s}$ samples of observation of the estimated parameter  with respect to the real parameter. As in the previous case, the estimation using the DF approach is accurate and the order of the ratio is the unity $\displaystyle \mathcal{O}\,\biggl(\, \frac{k_{\,\circ}}{k_{\,\real}} \,\biggr) \, \simeq \, 1 \,$. A slight increase of the standard deviation with the thermal conductivity can be noted in Table~\ref{tab:k_results_stats}. For the RC model, the estimation is not satisfactory. The maximum error goes upt to $80 \%$ and is observed for large thermal conductivity $k_{\,\real} \egal 2.5 \ \mathsf{W\,/\,(m^{\,2} \cdot K)}\,$. In addition, the estimation error is decreasing faster than for the previous case with a slope around $\simeq \, - \,0.18\,$. 
The number of iterations to estimate the parameter is stable around $8$ for the algorithm using the DF model. For the RC approach, the algorithm needs more iterations. The number of iterations tends to increase with the thermal conductivity.
Figure~\ref{fig:k_cpu_flamb} gives the mean of the computational time required by the algorithm to estimate the unknown parameter $k\,$. More details are provided in Table~\ref{tab:k_results_stats}. As expected, the approach using DF is longer. As the number of iterations to converge increases with $k_{\,\real}$ for the RC approach, the ratio of CPU times between both models decreases.

A comparison between the prediction of the models, computed with the estimated parameter $k_{\,\circ}\,$, and the experimental observations is shown in Figures~\ref{fig:k_T_all_ft_mat2} and ~\ref{fig:k_T_all_ft_mat5} for materials $2$ and $5\,$, respectively. The predictions obtained with the RC model are not reliable. The difference between the observations and the RC numerical predictions can reach $2 \ \mathsf{^{\,\circ}C}\,$. Figures~\ref{fig:k_crit1_fk} and \ref{fig:k_crit2_fk} present the variation of the convergence criteria with the number of iterations for material $5\,$. It is consistent with results presented in Figure~\ref{fig:k_Nk_flamb}. The algorithm based on RC model requires more iterations to converge. It can be remarked that for $8$ iterations, in the algorithm using the DF model, both criteria $\gamma_{\,1}$ and $\gamma_{\,2}$ are satisfied. It indicates that both magnitudes of changes in the cost function and in the unknown parameter are low. For the algorithm with the RC model, only criteria $\gamma_{\,2}$ on the magnitude of the cost function is satisfied for $12$ iterations.

\begin{figure}
\centering
\includegraphics[width=.45\textwidth]{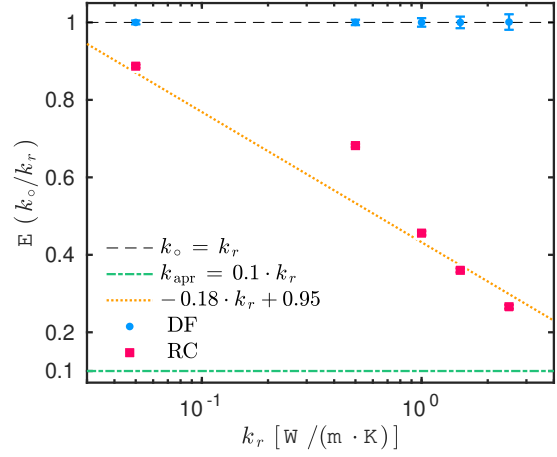}
\caption{Comparison of the expectation of the estimated parameter $k_{\,\circ}$ with the real parameter $k_{\,\real}\,$.}
\label{fig:k_P_flamb}
\end{figure}

\begin{figure}
\centering
\subfigure[\label{fig:k_Nk_flamb}]{\includegraphics[width=.45\textwidth]{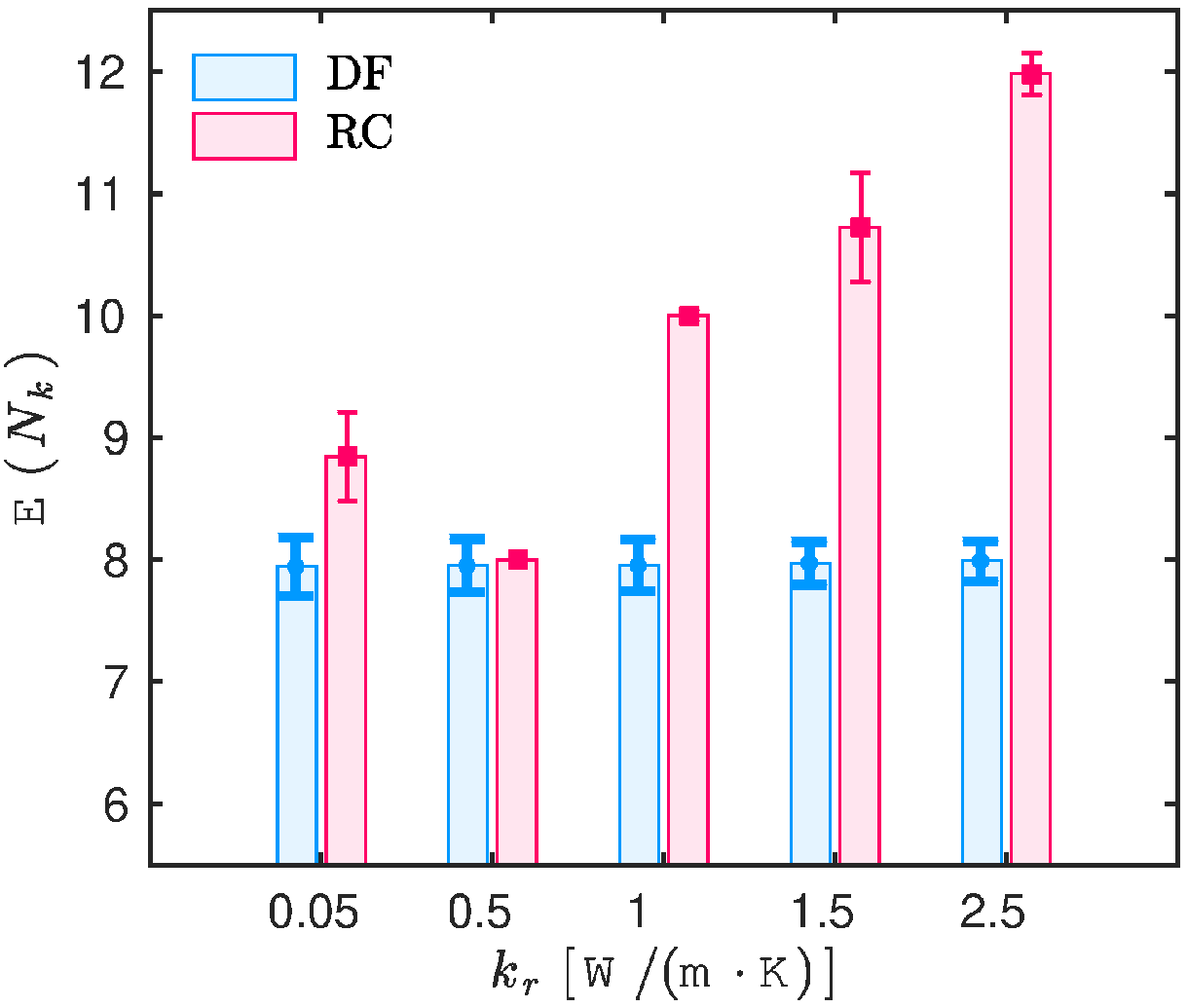}}  \hspace{0.2cm}
\subfigure[\label{fig:k_cpu_flamb}]{\includegraphics[width=.45\textwidth]{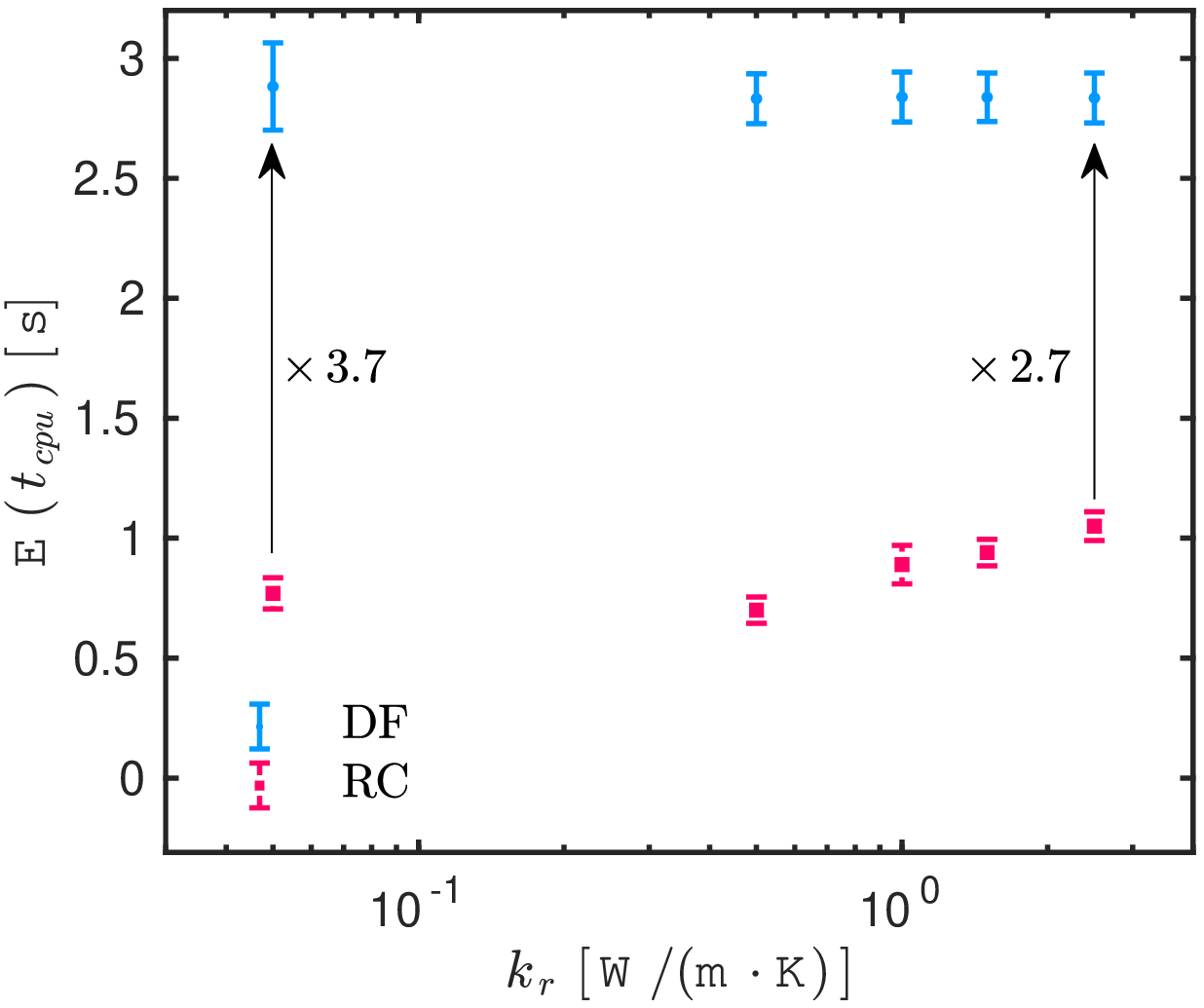}}
\caption{Variation of the expectation of the number of iteration $N_{\,m}$ \emph{(a)} and the computational time $t_{\,\mathrm{cpu}}$ \emph{(b)} for the algorithm to estimate the unknown parameter $k$ for the five types of material.}
\end{figure}

\begin{figure}
\centering
\subfigure[\label{fig:k_T_all_ft_mat2}]{\includegraphics[width=.45\textwidth]{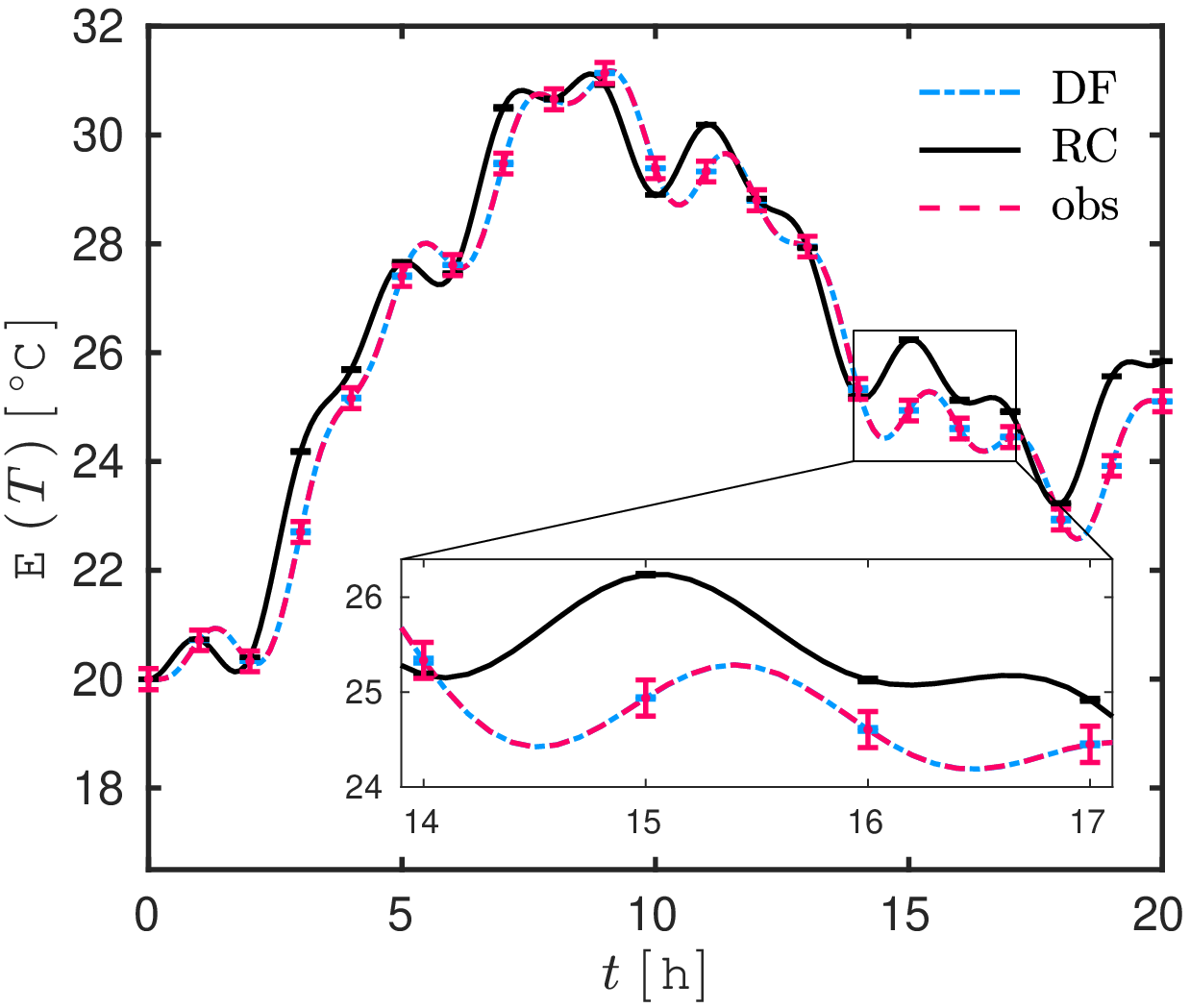}}  \hspace{0.2cm}
\subfigure[\label{fig:k_T_all_ft_mat5}]{\includegraphics[width=.45\textwidth]{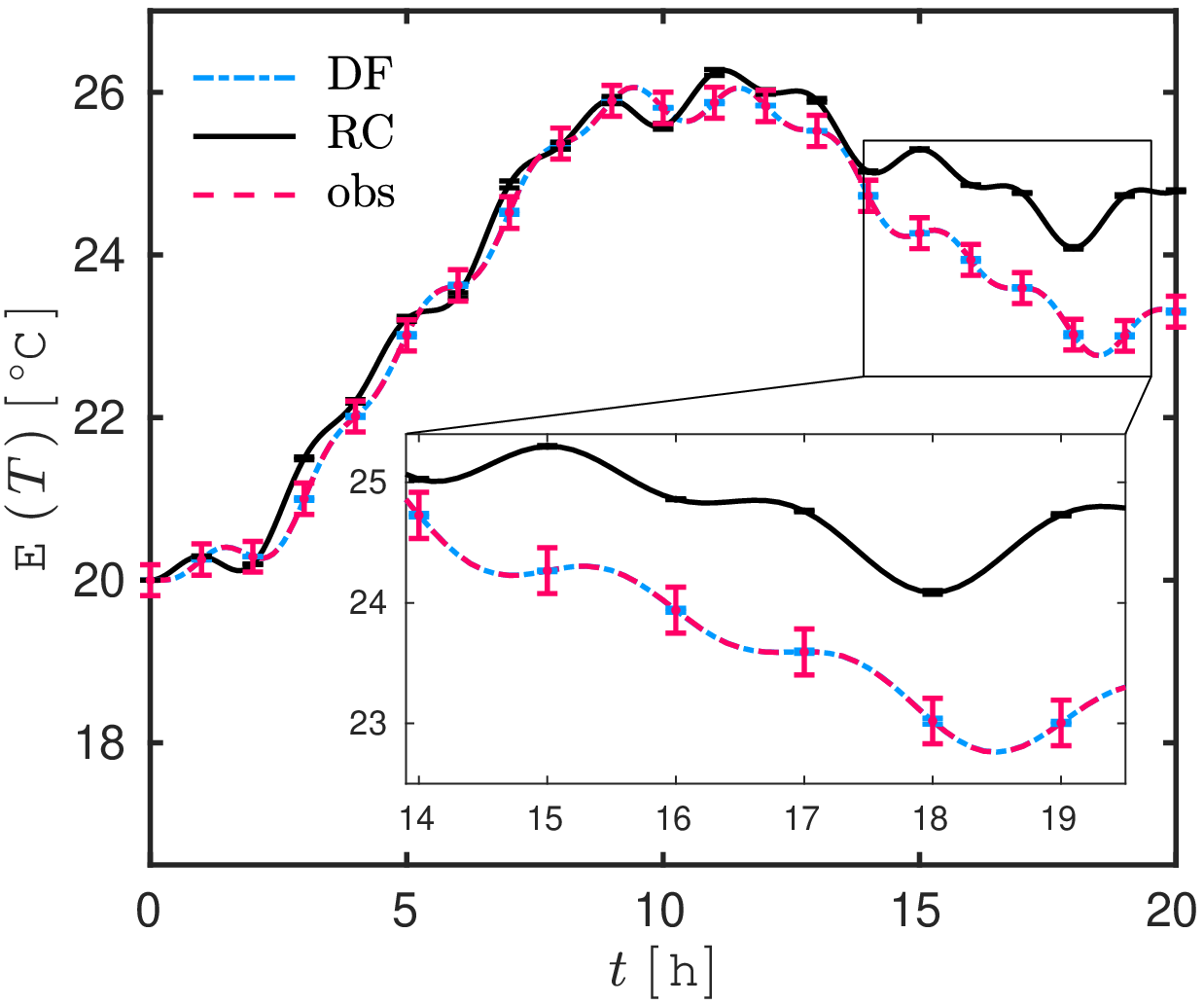}}
\caption{Time evolution of the temperature at $x \egal x_{\,\obs} \egal 11 \ \mathsf{cm}$ for material $2$ \emph{(a)} and material $5$ \emph{(b)} computed with the numerical model for $k \egal k_{\,\circ}\,$.}
\end{figure}

\begin{figure}
\centering
\subfigure[\label{fig:k_crit1_fk}]{\includegraphics[width=.45\textwidth]{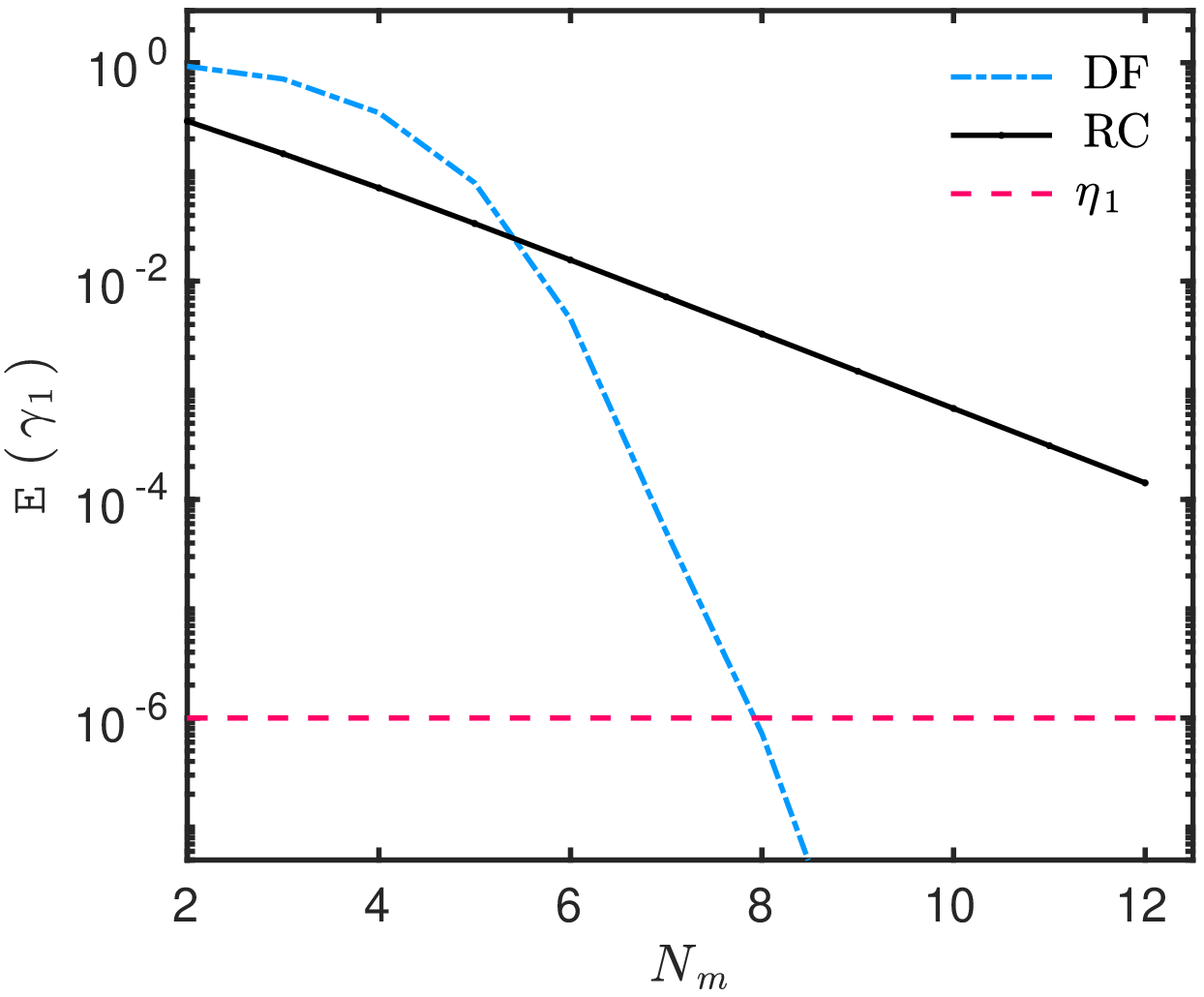}}  \hspace{0.2cm}
\subfigure[\label{fig:k_crit2_fk}]{\includegraphics[width=.45\textwidth]{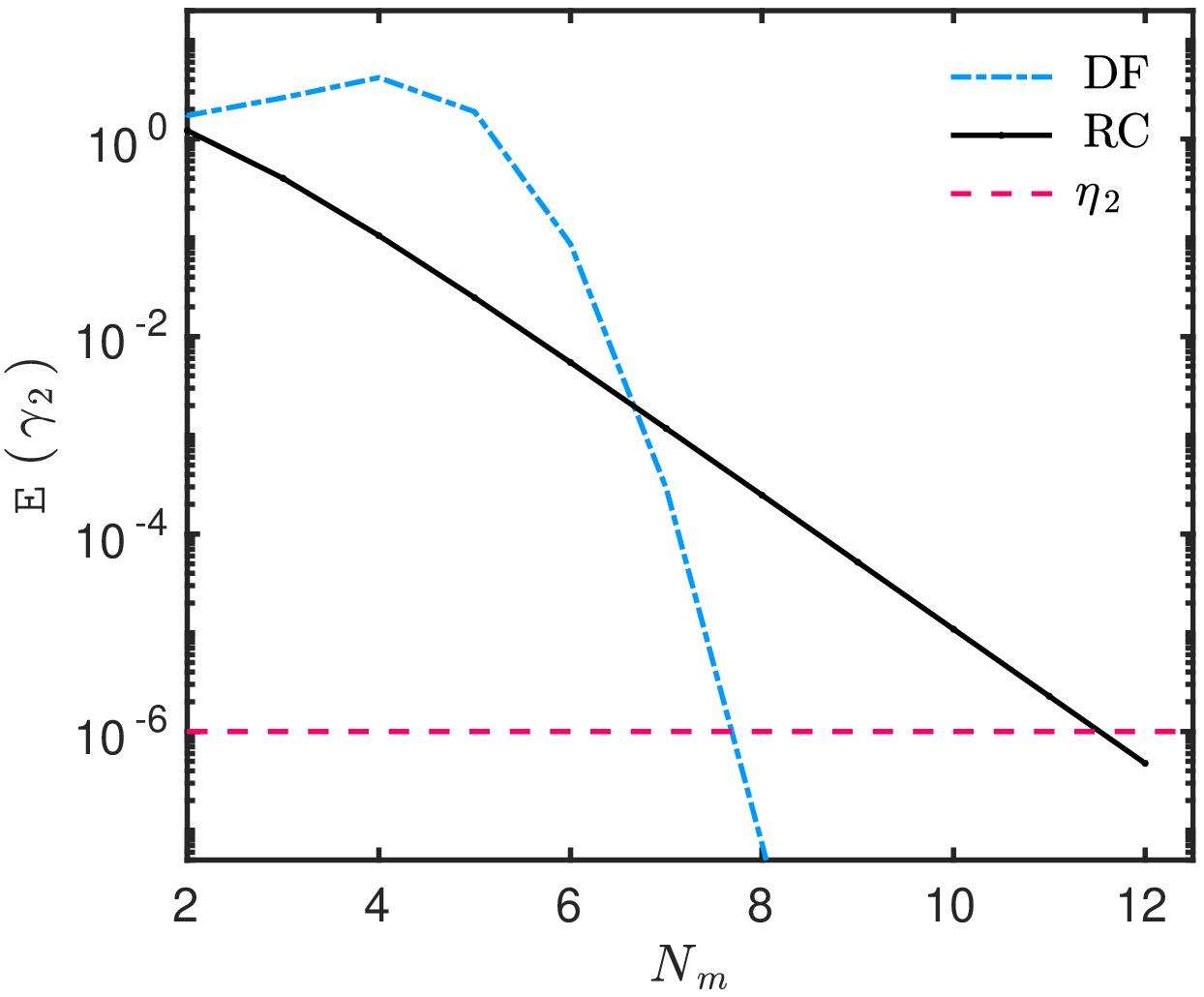}}
\caption{Variation of the expectation of the convergence criteria $\gamma_{\,1}$ \emph{(a)} and $\gamma_{\,2}$ \emph{(b)} for the estimation of the unknown parameter $k$ for the material $5$.}
\end{figure}

\begin{table}
\centering
\caption{Results for the estimation of the unknown thermal conductivity $k \,$.}
\label{tab:k_results_stats}
\setlength{\extrarowheight}{.5em}
\begin{tabular}[l]{@{} c|cc|cc|cc|cc|cc|cc}
\hline
\hline
& \multicolumn{4}{c|}{\textit{Ratio $\displaystyle \frac{k_{\,\circ}}{k_{\,\real}}$}}
& \multicolumn{4}{c|}{\textit{Number of iterations $N_{\,m}$}}
& \multicolumn{4}{c}{\textit{Computational time $t_{\,\mathrm{CPU}} \ \unit{s}$}} \\[8pt]
\textit{Material}
& \multicolumn{2}{c|}{\textit{DF model}}
& \multicolumn{2}{c|}{\textit{RC model}}
& \multicolumn{2}{c|}{\textit{DF model}}
& \multicolumn{2}{c|}{\textit{RC model}}
& \multicolumn{2}{c|}{\textit{DF model}}
& \multicolumn{2}{c}{\textit{RC model}} \\
\textit{Identification}
& $\mathsf{E}$
& $\sigma$
& $\mathsf{E}$
& $\sigma$
& $\mathsf{E}$
& $\sigma$
& $\mathsf{E}$
& $\sigma$
& $\mathsf{E}$
& $\sigma$
& $\mathsf{E}$
& $\sigma$ \\
\hline
1
& $1.0$ & $0.005$
& $0.89$ & $0.004$
& $7.9$ & $0.24$
& $8.8$ & $0.36$ 
& $2.8$ & $0.18$
& $0.7$ & $0.05$ \\
2
& $1.0$ & $0.007$
& $0.68$ & $0.004$
& $7.9$ & $0.22$
& $8.0$ & $0$ 
& $2.8$ & $0.10$
& $0.7$ & $0.05$ \\
3
& $1.0$ & $0.011$
& $0.46$ & $0.005$
& $7.9$ & $0.21$
& $10.0$ & $0.04$ 
& $2.8$ & $0.10$
& $0.9$ & $0.1$ \\
4
& $1.0$ & $0.015$
& $0.36$ & $0.005$
& $8.0$ & $0.18$
& $10.7$ & $0.45$ 
& $2.8$ & $0.10$
& $0.9$ & $0.05$ \\
5
& $1.0$ & $0.02$
& $0.26$ & $0.005$
& $8.0$ & $0.16$
& $12$ & $0.17$ 
& $2.8$ & $0.10$
& $0.9$ & $0.05$ \\
\hline
\hline
\end{tabular}
\end{table}

\subsection{Estimation of the heat transfer coefficient}

The last case study concerns the estimation of the heat transfer coefficient $h_{\,L}\,$. The identifiability of this parameter is first proven for the DF model. 
\begin{proposition}
The parameter $h_{\,L}$ is identifiable in Equation~\eqref{eq:BC_L}.
\end{proposition}
\begin{proof}
We assume an observable $T\,\bigl(\,-\,,\,-\bigr)$ obtained for the model:
\begin{align}
\label{eq:SGI_heat_eq_h}
k \cdot \pd{T}{n} \egal h_{\,L} \cdot \Bigl(\, T \moins T_{\,\infty\,,\,L} \,\Bigr)\,.
\end{align}
Another observable, denoted with a superscript $^{\,\prime}\,$, obtained with another parameter $k^{\,\prime}$ is detained:
\begin{align}
\label{eq:SGI_heat_eq_h_prime}
k \cdot \pd{T^{\,\prime}}{n} \egal h_{\,L}^{\,\prime} \cdot \Bigl(\, T^{\,\prime} \moins T_{\,\infty\,,\,L} \,\Bigr)\,.
\end{align}
If $T \, \equiv \, T^{\,\prime}$ then $\displaystyle \pd{T}{n} \, \equiv \, \pd{T^{\,\prime}}{n}\,$. Thus, from Equations~\eqref{eq:SGI_heat_eq_h} and \eqref{eq:SGI_heat_eq_h_prime}, we obtain:
\begin{align*}
\Bigl(\, h_{\,L} \moins h_{\,L}^{\,\prime} \,\Bigr) \cdot T \egal 0 \,.
\end{align*}
Thus, $h_{\,L} \, \equiv \, h_{\,L}^{\,\prime}$ and parameter $h_{\,L}$ is SGI. 
\end{proof}
Secondly, the identifiability is proven for the RC model. \\~
\begin{proposition}
The parameter $h_{\,L}$ is identifiable in Equation~\eqref{eq:RC_model_BC}.
\end{proposition}
\begin{proof}
We assume an observable $T\,\bigl(\,-\,,\,-\,\bigr)$ obtained from the RC model:
\begin{align}
\label{eq:SGI_RC_model_h}
\frac{k}{e} \cdot \Bigl(\, T_{\,2} \moins T_{\,1} \,\Bigr) &  \egal h_{\,L} \cdot \Bigl(\, T_{\,1} \moins T_{\,\infty\,,\,L} \,\Bigr) \,.
\end{align}
Another observable, denoted with a superscript $^{\,\prime}\,$, obtained with another parameter $k^{\,\prime}$ holds:
\begin{align}
\label{eq:SGI_RC_model_h_prime}
\frac{k}{e} \cdot \Bigl(\, T_{\,2}^{\,\prime} \moins T_{\,1}^{\,\prime} \,\Bigr)&  \egal h_{\,L}^{\,\prime} \cdot \Bigl(\, T_{\,1}^{\,\prime} \moins T_{\,\infty\,,\,L} \,\Bigr) \,.
\end{align}
If $T \, \equiv \, T^{\,\prime}$ then from Equations~\eqref{eq:SGI_RC_model_h} and \eqref{eq:SGI_RC_model_h_prime}, one obtain:
\begin{align*}
\Bigl(\, h_{\,L} \moins h_{\,L}^{\,\prime} \,\Bigr) \cdot T_{\,1} 
\end{align*}
and $h_{\,L} \, \equiv \, h_{\,L}^{\,\prime}$ and parameter $h_{\,L}$ is SGI in the RC model. 
\end{proof}

The properties are the one from material $3$ identified in Table~\ref{tab:material}. The $N_{\,s}$ samples of observations are generated for four cases identified in Table~\ref{tab:heat_coefficient}. The chosen real values for $h_{\,L}$ corresponds to classical one encountered in the literature for building physics applications. The initial guess used in the algorithm to estimate the unknown parameter is $h_{\,L\,,\,\apr} \egal 0.1 \cdot h_{\,L\,,\,\real} \,$.

The expectation of the estimated parameter is shown in Figure~\ref{fig:h_P_fh} for the four cases. Again, the estimation performed with the DF model is accurate. The expectation of the ratio $\displaystyle \frac{h_{\,L\,,\,\circ}}{h_{\,L\,,\,\real}}$ scales with $1$ for all values of $h_{\,L\,,\,\real}\,$. When using the RC model, the estimation lacks of reliability. Particularly, for small values of heat transfer coefficient $h_{\,L\,,\,\real} \egal 0.5 \ \mathsf{W\,/\,(m^{\,2}\cdot K)}\,$, the estimated parameter is more than five time higher than the real one. For higher values of heat transfer coefficient $h_{\,L\,,\,\real} \egal 0.5 \ \mathsf{W\,/\,(m^{\,2}\cdot K)}\,$, the error on the estimation reaches $\simeq \, 30 \% \,$. 
Figures~\ref{fig:h_Nk_fh} and \ref{fig:h_cpu_fh} show the variation of the number of iterations and the CPU time for the algorithms to estimate the unknown parameter. The number of iterations remains stable for the RC approach while it increases for the DF one. As expected, the DF model has a higher computational cost, increasing the time required to solve the parameter estimation problem. Detailed results are also provided in Table~\ref{tab:heat_coefficient}. Compared to previous case study, the computational cost is divided by two for the algorithm using the DF model. Indeed, the algorithm requires fewer iterations to converge. 

The comparison between the numerical predictions and the experimental observations is provided in Figures~\ref{fig:h_T_all_ft_h1} and \ref{fig:h_T_all_ft_h3}. Important discrepancies are noted for the predictions using the RC model with the estimated heat transfer coefficient $h_{\,L\,,\,\circ}\,$. For the case $1\,$, the error can reach $\simeq \, 4 \ \mathsf{^{\,\circ}C}\,$. For lower values of heat transfer coefficient $h_{\,L\,,\,\real} \egal 0.5 \ \mathsf{W\,/\,(m^{\,2}\cdot K)}\,$, the boundary conditions at $x \egal 0$ tends to be adiabatic. The RC model is completely unreliable to predict the physical phenomena for such cases. For larger heat transfer coefficient values, the discrepancies are lower but the predictions of the model are still not satisfactory. The speed of convergence of the algorithm is illustrated in Figures~\ref{fig:h_crit1_fk} and \ref{fig:h_crit2_fk}. As for the previous cases, only the criteria $\gamma_{\,2}$ on the magnitude of the changes in the cost function is reached for the algorithm using the RC approach. It is another indication of the poor accuracy of the method.

\begin{table}
\centering
\caption{Results for the estimation of the unknown thermal conductivity $h_{\,L} \,$.}
\label{tab:heat_coefficient}
\setlength{\extrarowheight}{.5em}
\begin{tabular}[l]{@{} c|c|cc|cc|cc|cc|cc|cc}
\hline
\hline
\textit{Case}
& \textit{Real value} 
& \multicolumn{4}{c|}{\textit{Ratio $\displaystyle \frac{h_{\,L\,,\,\circ}}{h_{\,L\,,\,\real}}$}}
& \multicolumn{4}{c|}{\textit{Number of iterations $N_{\,m}$}}
& \multicolumn{4}{c}{\textit{Computational time $t_{\,\mathrm{CPU}} \ \unit{s}$}} \\[8pt]
& $h_{\,L\,,\,\real} \ \unitt{W}{m^{\,2} \cdot K}$
& \multicolumn{2}{c|}{\textit{DF model}}
& \multicolumn{2}{c|}{\textit{RC model}}
& \multicolumn{2}{c|}{\textit{DF model}}
& \multicolumn{2}{c|}{\textit{RC model}}
& \multicolumn{2}{c|}{\textit{DF model}}
& \multicolumn{2}{c}{\textit{RC model}} \\
1 
& $0.5$
& $1.0$ & $0.07$
& $5.5$ & $0.06$
& $4$ & $0.06$
& $6$ & $0$ 
& $1.1$ & $0.08$
& $0.7$ & $0.05$ \\
2 
& $5$ 
& $1.0$ & $0.01$
& $1.05$ & $0.01$
& $5$ & $0.07$
& $6$ & $0$ 
& $1.1$ & $0.08$
& $0.6$ & $0.05$ \\
3 
& $10$
& $1.0$ & $0.01$
& $0.82$ & $0.01$
& $5.9$ & $0.32$
& $6$ & $0$ 
& $1.3$ & $0.08$
& $0.6$ & $0.05$ \\
4 
& $15$ 
& $1.0$ & $0.01$
& $0.74$ & $0.06$
& $6$ & $0.13$
& $6.5$ & $0.5$ 
& $1.3$ & $0.05$
& $0.7$ & $0.05$\\
\hline
\hline
\end{tabular}
\end{table}

\begin{figure}
\centering
\includegraphics[width=.45\textwidth]{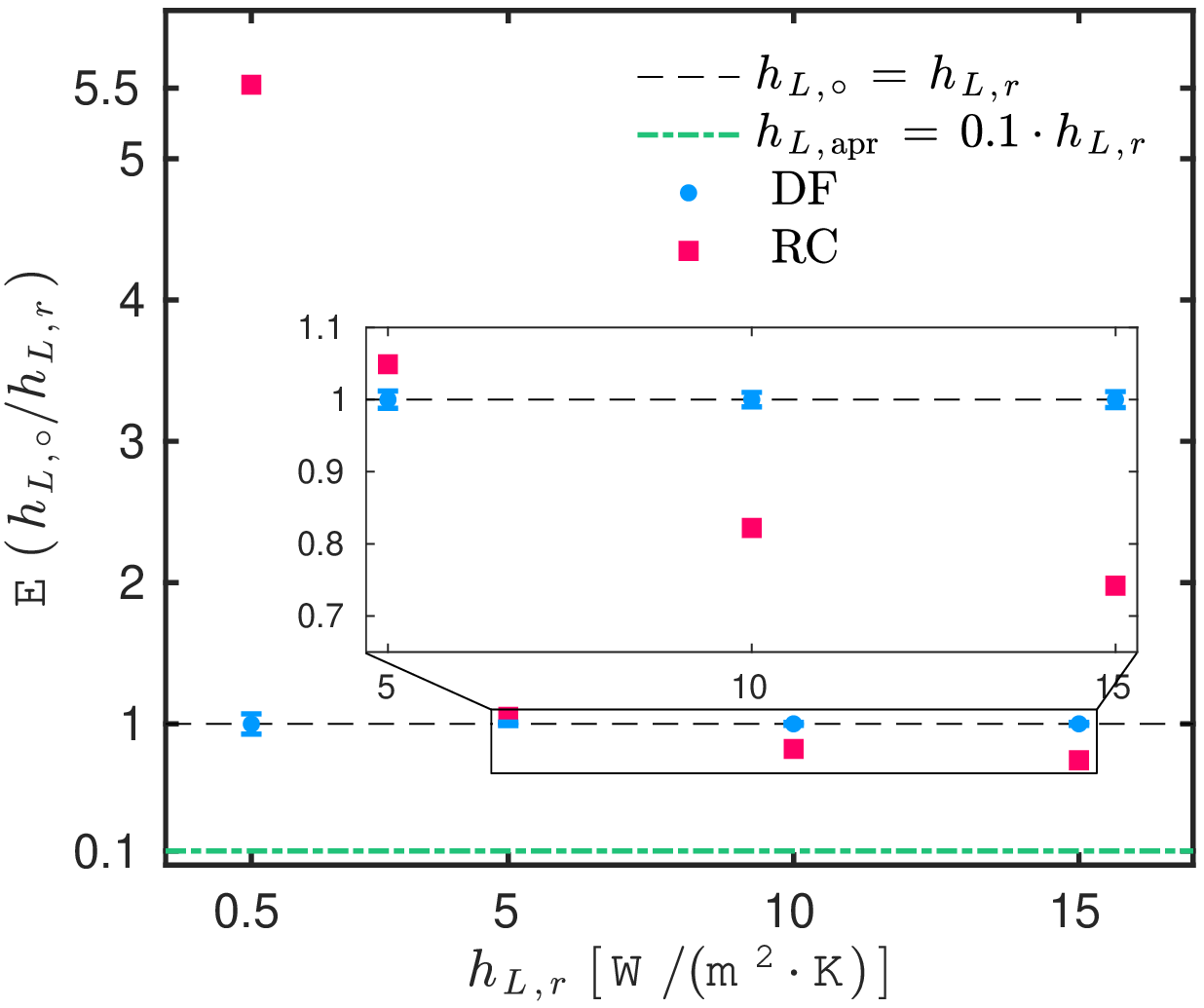}
\caption{Comparison of the expectation of the estimated parameter $h_{\,L\,,\,\circ}$ with the real parameter $h_{\,L\,,\,\real}\,$.}
\label{fig:h_P_fh}
\end{figure}

\begin{figure}
\centering
\subfigure[\label{fig:h_Nk_fh}]{\includegraphics[width=.45\textwidth]{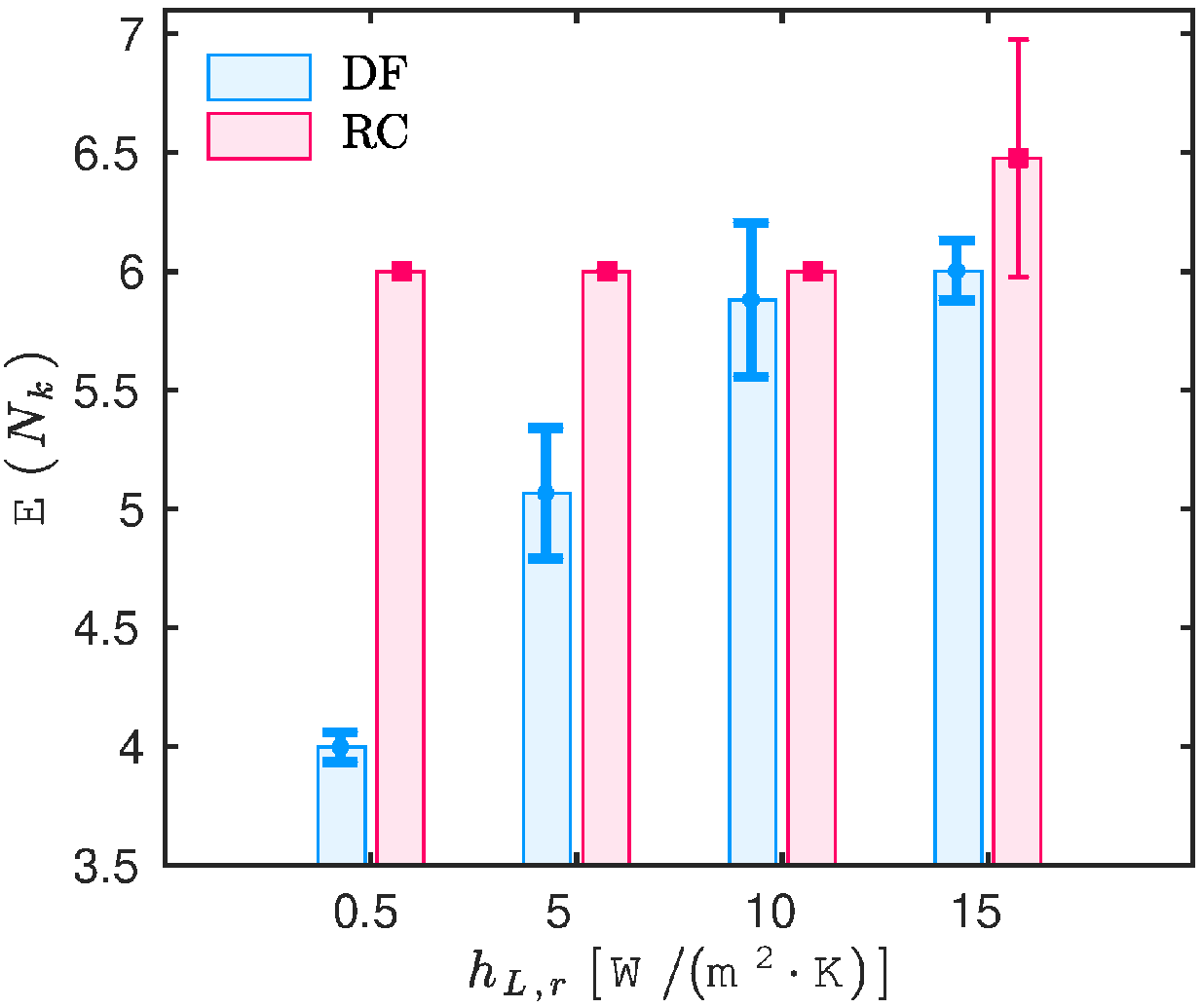}}  \hspace{0.2cm}
\subfigure[\label{fig:h_cpu_fh}]{\includegraphics[width=.45\textwidth]{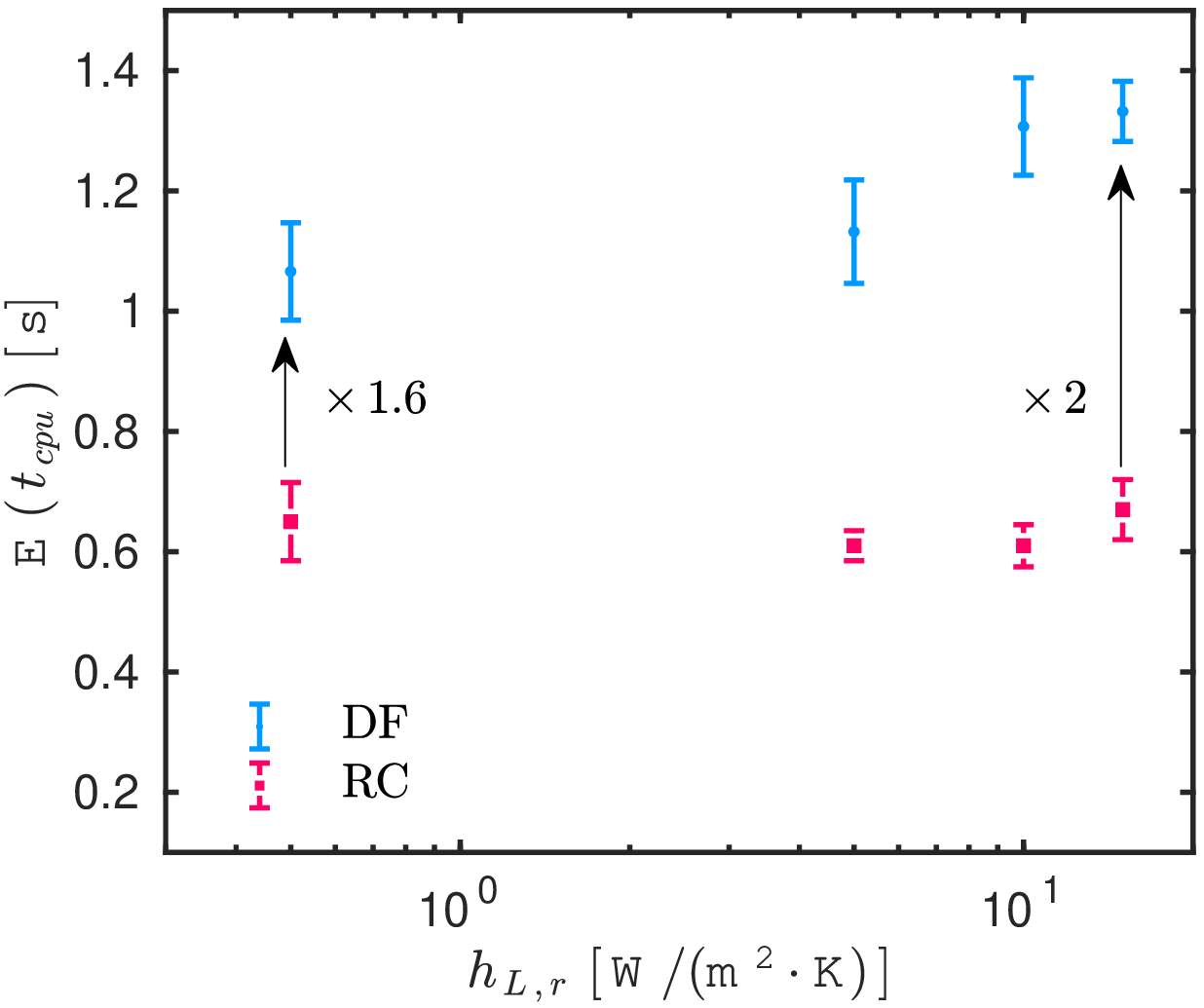}}
\caption{Variation of the expectation of the number of iteration $N_{\,m}$ \emph{(a)} and the computational time $t_{\,\mathrm{cpu}}$ \emph{(b)} for the algorithm to estimate the unknown parameter $h_{\,L}$ for the four cases.}
\end{figure}

\begin{figure}
\centering
\subfigure[\label{fig:h_T_all_ft_h1}]{\includegraphics[width=.45\textwidth]{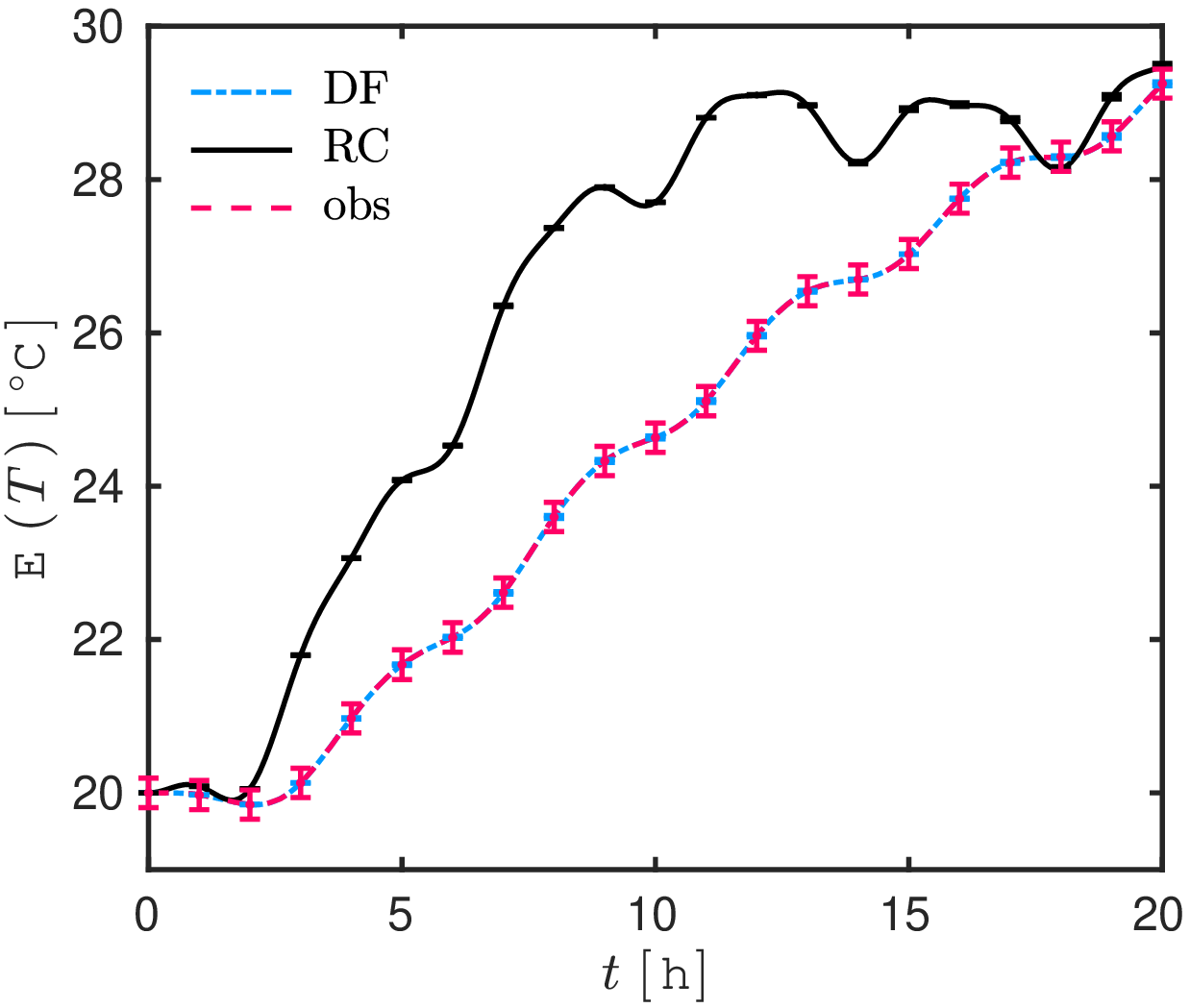}}  \hspace{0.2cm}
\subfigure[\label{fig:h_T_all_ft_h3}]{\includegraphics[width=.45\textwidth]{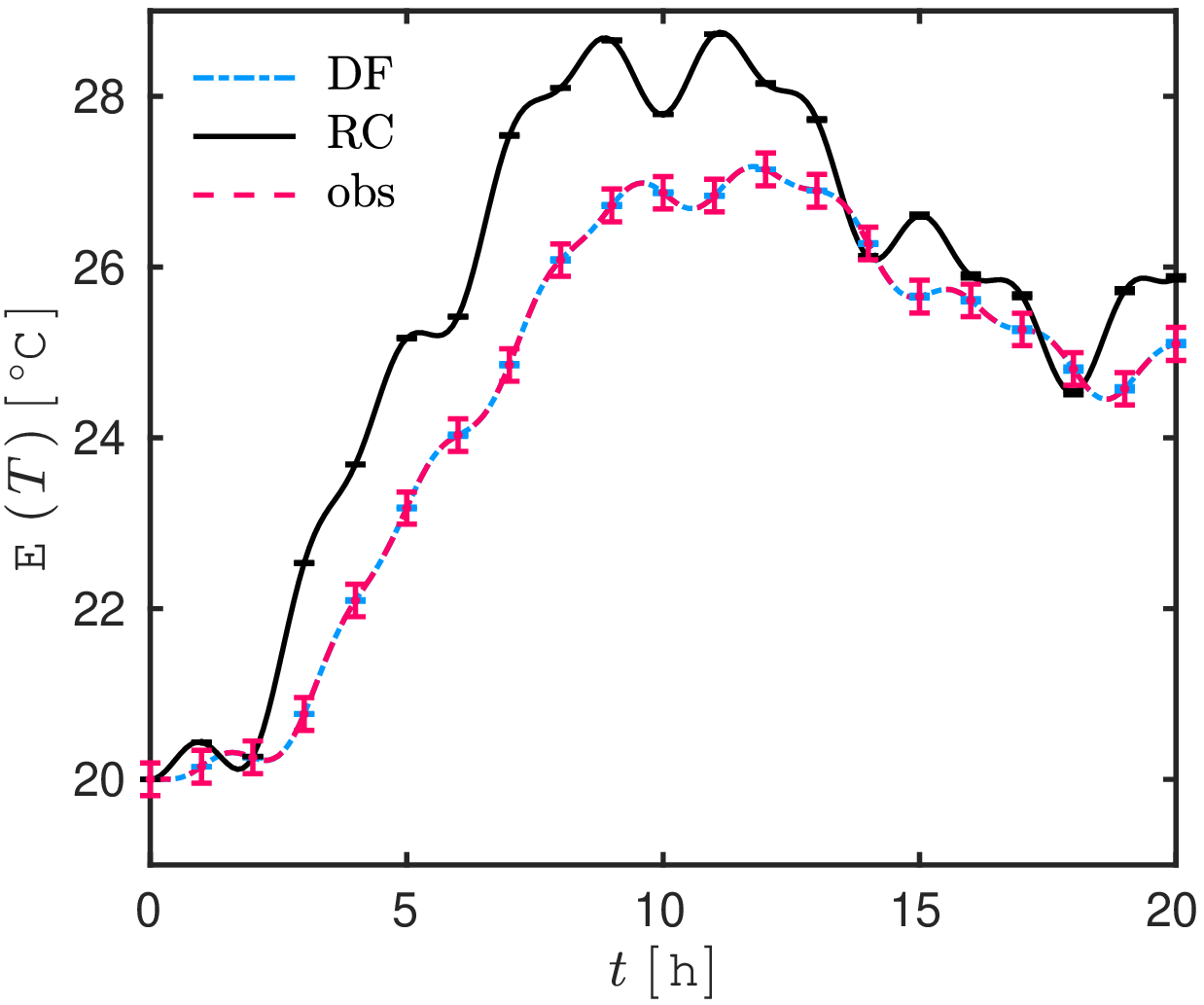}}
\caption{Time evolution of the temperature at $x \egal x_{\,\obs} \egal 11 \ \mathsf{cm}$ for cases $1$ \emph{(a)} and $3$ \emph{(b)} computed with the numerical model for $h_{\,L} \egal h_{\,L\,,\,\circ}\,$.}
\end{figure}

\begin{figure}
\centering
\subfigure[\label{fig:h_crit1_fk}]{\includegraphics[width=.45\textwidth]{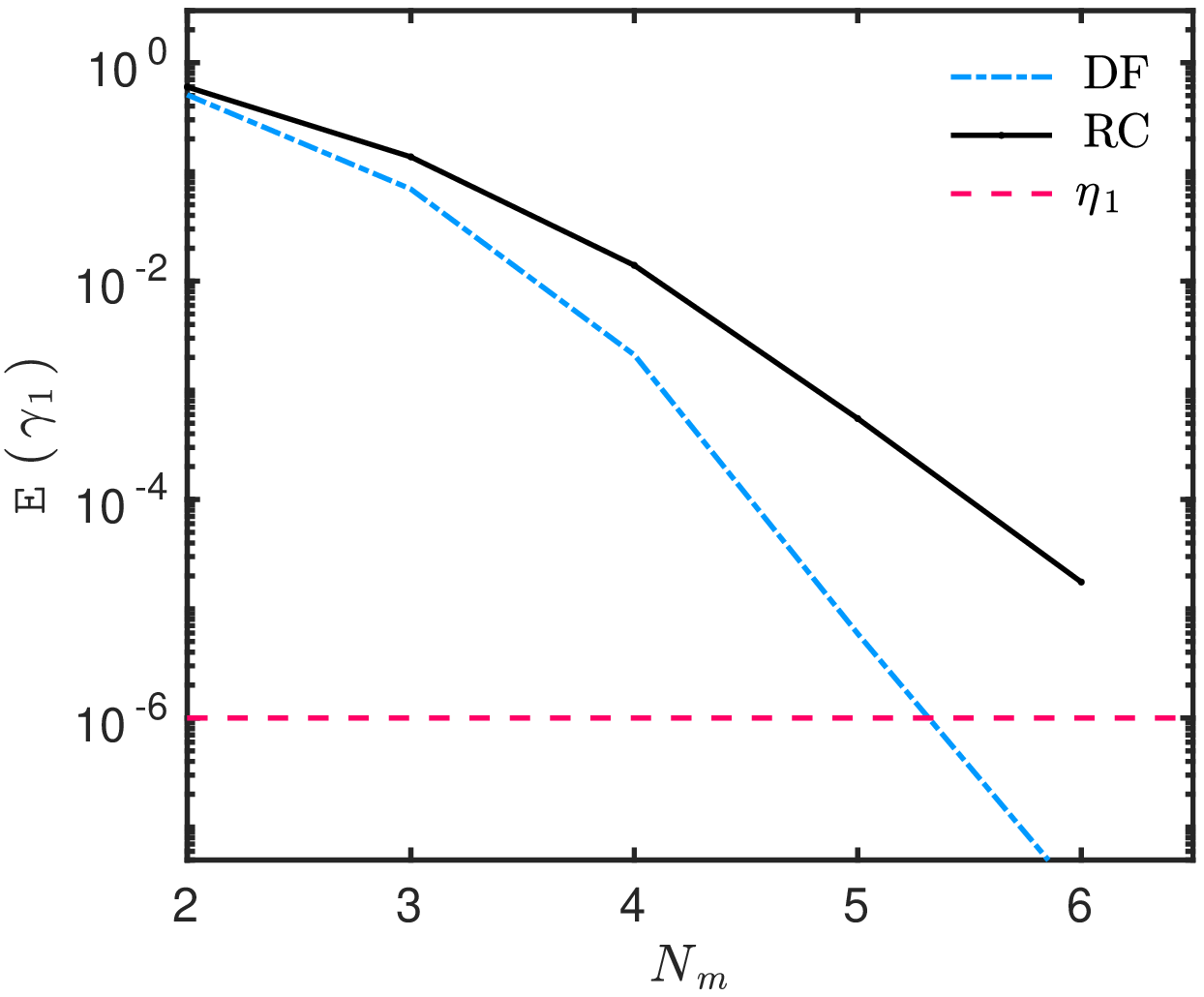}}  \hspace{0.2cm}
\subfigure[\label{fig:h_crit2_fk}]{\includegraphics[width=.45\textwidth]{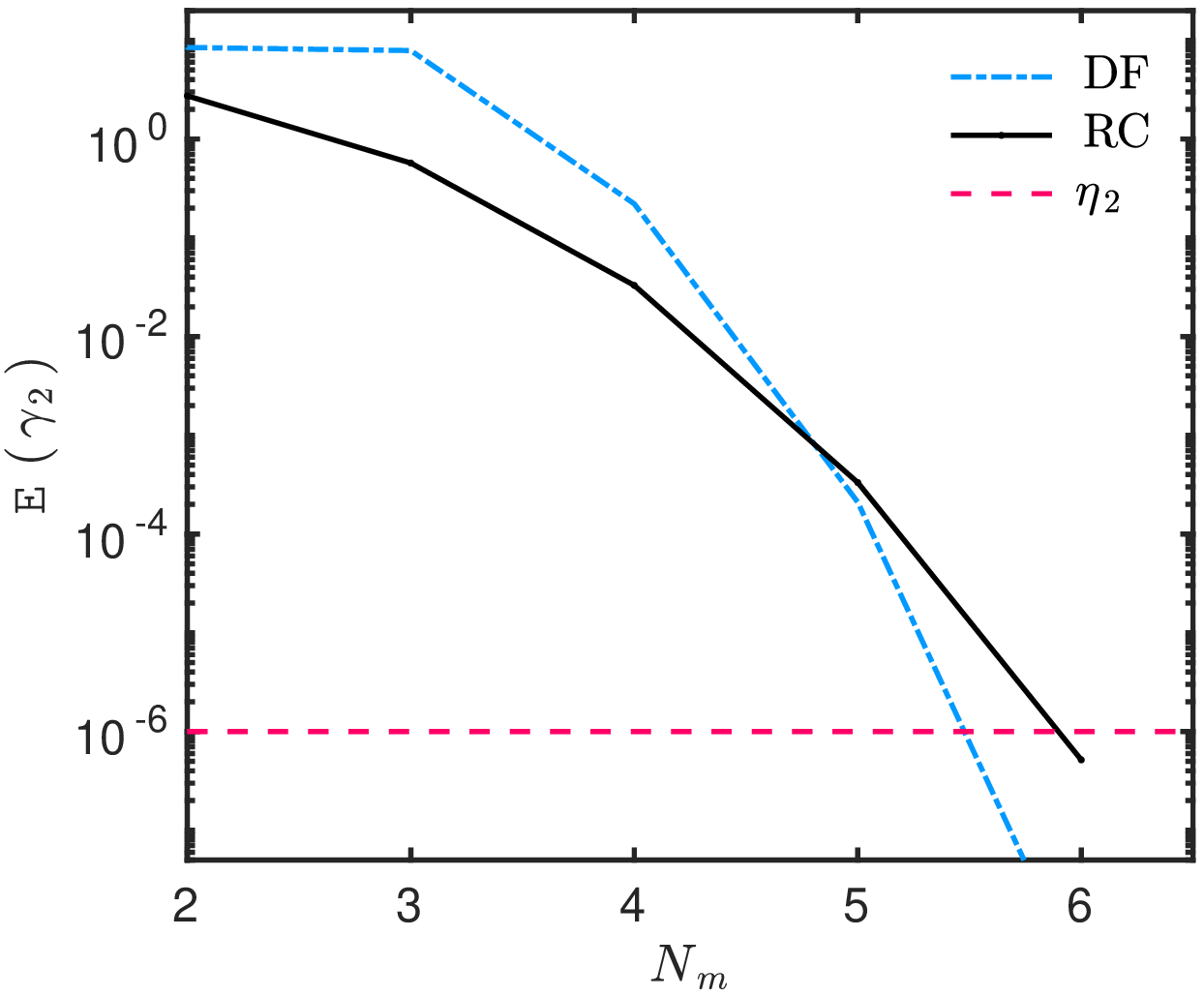}}
\caption{Variation of the expectation of the convergence criteria $\gamma_{\,1}$ \emph{(a)} and $\gamma_{\,2}$ \emph{(b)} for the estimation of the unknown parameter $h_{\,L}$ for the case $3$.}
\end{figure}

\section{Conclusion}
\label{sec:conclusion}

In building physics, it is of capital importance to build reliable models to simulate the physical phenomena of heat losses through the walls. The fidelity of a model can be evaluated by comparing the numerical predictions with experimental observations. The reliability can also be assessed by the robustness of the model to estimate accurate unknown parameters using given observations. This article deals with this second aspect of reliability. 

Two main mathematical models are proposed in the literature for heat losses through building wall. The first one, denoted by DF, is based on the heat diffusion equation combined with the \DF ~scheme to build the numerical model. The second one is the so-called lumped RC model which approximates the diffusion processes by an ordinary differential equation. Within this approach, only three temperatures are evaluated in the wall. Section~\ref{sec:mathematical_model} presents the two mathematical models used to obtain the solution of the direct problem when estimating the parameter. Then, in Section~\ref{sec:evaluating_reliability}, the methodology to evaluate the reliability is detailed. First, samples of observations are generated using a pseudo-spectral numerical method for the heat diffusion equation and a known parameter. A noise is then added to the numerical results to generate experimental observations \emph{in silico}. The second step consists in solving the parameter estimation problem for each sample of observations using both mathematical model. The main criterion to evaluate the reliability is the accuracy of recovering the unknown parameter. Secondary criteria focus on computational time and number of iterations to solve the inverse problem. 

In Section~\ref{sec:case_studies}, three case studies are considered for the estimation of (\emph{i}) the heat capacity, (\emph{ii}) the thermal conductivity or (\emph{iii}) the heat transfer coefficient at the interface between the wall and the ambient air. For each case, a total of $10^{\,4}$ samples of observations are generated. The parameter estimation problem is then solved with each mathematical model. The results highlight a very satisfactory robustness of the DF approach to estimate the unknown parameter. For each case, the parameter is recovered with $100 \%$ accuracy. On the other hand, the reliability of the RC model is not satisfactory. For the estimation of the heat capacity or the thermal conductivity, the error can reach $40 \%$ or $80\% \,$, respectively. For the estimation of the heat transfer coefficient at the interface between the ambient air and the material, the relative error goes up to $450 \%$ for small magnitude of the coefficient. The accuracy of the estimation is unacceptable for the RC approach revealing a lack of reliability in the mathematical model. For all cases, the algorithm using the DF approach has a higher computational time even if it requires less iterations to converge. The computational cost of the DF model is a reasonable price to pay to have a reliable model to estimate parameters. Beyond these results, it should be recalled the importance of having confidence in the estimated thermal conductivity $k\,$, heat capacity $c\,$ or surface heat transfer coefficient. Once estimated, these parameters are used in computational tools to perform direct simulations and evaluate building energy efficiency, particularly in the context of thermal regulations. Thus, if the parameters values are not reliable, the predictions of building energy requirements might be inaccurate.

As a conclusion, the choice of the mathematical model to simulate the heat losses through a building wall has to be done carefully to be able to rely later on this model predictions. Further studies will investigate the reliability of more complex mathematical models involving coupled heat and mass transfers. Indeed, the latent effects impact strongly the prediction of the building energy efficiency as it was demonstrated in \cite{Mendes_2017}.

\section*{Nomenclature}

\begin{tabular}[l]{@{} cll}
\hline
\hline
\multicolumn{3}{c}{\emph{Latin letters}} \\
$\Bi$ & \textsc{Biot} number & $\unit{\emptyset}$ \\
$c$ & specific heat capacity & $\unitt{W}{kg \cdot K}$ \\
$e$ & length & $\unit{m}$ \\
$\Fo$ & \textsc{Fourier} number & $\unit{\emptyset}$ \\
$h$ & surface heat transfer coefficient & $\unitt{W}{m^{\,2} \cdot K}$ \\
$k$ & thermal conductivity & $\unitt{W}{m \cdot K}$ \\
$N$ & number & $\unit{\emptyset}$ \\
$L$ &  length & $\unit{m}$ \\
$t$ & time coordinate & $\unit{s}$ \\
$T$ & temperature & $\unit{K}$ \\
$u$ & dimensionless temperature & $\unit{\emptyset}$ \\
$x$ & space coordinate & $\unit{m}$ \\
\hline
\hline
\end{tabular}

\hspace{2cm} \\~

\begin{tabular}[l]{@{} cl}
\hline
\hline
\multicolumn{2}{c}{\emph{Subscript or super script}} \\
$\infty$ & ambient air \\
$\circ$ & estimated parameter \\
$\star$ & dimensionless variable \\
$\apr$ & \emph{a priori} value parameter \\
$\obs$ & experimental observation \\
$\real$ & real parameter \\
$\reff$ & reference value \\
\hline
\hline
\end{tabular}

\section*{Acknowledgments}

The authors acknowledge the Junior Chair Research program ``Building performance assessment, evaluation and enhancement'' from the University of Savoie Mont Blanc in collaboration with The French Atomic and Alternative Energy Center (CEA) and Scientific and Technical Center for Buildings (CSTB). The authors also would like to acknowledge Dr. Celine Labart (LAMA UMR 5127, University Savoie Mont Blanc, France) for her precious discussions on numerical matters. 

\bibliographystyle{unsrt}  
\bibliography{references}

\begin{thebibliography}{10}

\bibitem{Fourier_1822}
J.~Fourier.
\newblock {\em Théorie Analytique de la Chaleur}.
\newblock {J. Gabay}, Sceaux, France, 1988.

\bibitem{Kirkpatrick_1943}
S.~Kirkpatrick nad J.~Lee, T.~Olive, H.~Batters, J.~Callaham, N.~Farquhar, and
  L.~Pope.
\newblock Complex heat transfer solved by electrical analogy.
\newblock {\em Chemical and Metallurgical Engineering}, 50(12):111--113, 1943.

\bibitem{Lawson_1953}
D.~I. Lawson and J.~H. McGuire.
\newblock The solution of transient heat-flow problems by analogous electrical
  networks.
\newblock {\em Proceedings of the Institution of Mechanical Engineers},
  167(1):275--290, 1953.

\bibitem{Robertson_1958}
A~.~F. Robertson and D.~Gross.
\newblock An electrical-analog method for transient heat-flow analysis.
\newblock {\em Journal of Research of the National Bureau of Standards},
  61(2):105--115, 1958.

\bibitem{Mendes_2017}
N.~Mendes, M.~Chhay, J.~Berger, and D.~Dutykh.
\newblock {\em Numerical methods for diffusion phenomena in building physics}.
\newblock {PUCPress}, Curitiba, Brazil, 2016.

\bibitem{Fraisse_2002}
G.~Fraisse, C.~Viardot, O.~Lafabrie, and G.~Achard.
\newblock Development of a simplified and accurate building model based on
  electrical analogy.
\newblock {\em Energy and Buildings}, 34(10):1017--1031, 2002.

\bibitem{Kampf_2007}
J.H. Kampf and D.~Robinson.
\newblock A simplified thermal model to support analysis of urban resource
  flows.
\newblock {\em Energy and Buildings}, 39(4):445--453, 2007.

\bibitem{Naveros_2015}
I.~Naveros and C.~Ghiaus.
\newblock Order selection of thermal models by frequency analysis of
  measurements for building energy efficiency estimation.
\newblock {\em Applied Energy}, 139(Supplement C):230--244, 2015.

\bibitem{Roels_2017}
S.~Roels, P.~Bacher, G.~Bauwens, S.~Castaño, Jiménez, and H.~Madsen.
\newblock On site characterisation of the overall heat loss coefficient:
  comparison of different assessment methods by a blind validation exercise on
  a round robin test box.
\newblock {\em Energy and Buildings}, 153:179--189, 2017.

\bibitem{Jimenez_2009}
M.J. Jimenez, B.~Porcar, and M.R. Heras.
\newblock Application of different dynamic analysis approaches to the
  estimation of the building component {U} value.
\newblock {\em Building and Environment}, 44(2):361--367, 2009.

\bibitem{Berger_2016}
J.~Berger, H.R.B. Orlande, N.~Mendes, and S.~Guernouti.
\newblock Bayesian inference for estimating thermal properties of a historic
  building wall.
\newblock {\em Building and Environment}, 106(Supplement C):327--339, 2016.

\bibitem{Berger_2018b}
J.~Berger, S.~Gasparin, D.~Dutykh, and N.~Mendes.
\newblock Accuracy of numerical methods applied to building energy performance.
\newblock {\em Building Simulation}, Submitted.

\bibitem{Kabanikhin_2008}
S.I. Kabanikhin.
\newblock Definitions and examples of inverse ill-posed problems.
\newblock {\em Journal of Inverse Problems and Ill -posed problems},
  16:317--357, 2008.

\bibitem{Kabanikhin_2011}
S.I. Kabanikhin.
\newblock {\em Inverse and ill-posed problems: theory and applications}.
\newblock Walter De Gruyter, Berlin, 2011.

\bibitem{Berger_2018a}
J.~Berger, S.~Gasparin, D.~Dutykh, and N.~Mendes.
\newblock On the solution of coupled heat and moisture transport in porous
  material.
\newblock {\em Transport in Porous Media}, 121(3):665--702, 2018.

\bibitem{Nayfeh_2000}
A.~Nayfeh.
\newblock {\em Perturbation Methods}.
\newblock Wiley VCH, New York, 2000.

\bibitem{Kahan_1979}
W.~Kahan and J.~Palmer.
\newblock On a proposed floating-point standard.
\newblock {\em {ACM} {SIGNUM} Newsletter}, 14:13--21, 1979.

\bibitem{Du_Fort_1953}
E.~C.~Du Fort and S.~P. Frankel.
\newblock Stability conditions in the numerical treatment of parabolic
  differential equations.
\newblock {\em Mathematical Tables and Other Aids to Computation},
  7(43):135--152, 1953.

\bibitem{Taylor_1970}
P.~J. Taylor.
\newblock The stability of the {D}ufort-{F}rankel method for the diffusion
  equation with boundary conditions involving space derivatives.
\newblock {\em The Computer Journal}, 13(1):1--92, 1970.

\bibitem{Gasparin_2017a}
S.~Gasparin, J.~Berger, D.~Dutykh, and N.~Mendes.
\newblock Stable explicit schemes for simulation of nonlinear moisture transfer
  in porous materials.
\newblock {\em Journal of Building Performance Simulation}, 11(2):129--144,
  2018.

\bibitem{Gasparin_2017b}
S.~Gasparin, J.~Berger, D.~Dutykh, and N.~Mendes.
\newblock An improved explicit scheme for whole-building hygrothermal
  simulation.
\newblock {\em Building Simulation}, 11(3):465--481, 2018.

\bibitem{Davies_2004}
M.~G. Davies.
\newblock {\em Building Heat Transfer}.
\newblock {J}ohn {W}iley \& {S}ons, West Sussex, 2004.

\bibitem{Deconinck_2016}
A.H. Deconinck and S.~Roels.
\newblock Comparison of characterisation methods determining the thermal
  resistance of building components from onsite measurements.
\newblock {\em Energy and Buildings}, 130(Supplement C):309--320, 2016.

\bibitem{Reynders_2014}
G.~Reynders, J.~Diriken, and D.~Saelens.
\newblock Quality of grey-box models and identified parameters as function of
  the accuracy of input and observation signals.
\newblock {\em Energy and Buildings}, 82(Supplement C):263--274, 2014.

\bibitem{Walter_1982}
E.~Walter and Y.~Lecourtier.
\newblock Global approaches to identifiability testing for linear and nonlinear
  state space models.
\newblock {\em Mathematics and Computers in Simulation}, 24(6):472--482, 1982.

\bibitem{Chebfun_2014}
T.~A. Driscoll, N.~Hale, and L.~N. Trefethen.
\newblock Chebfun guide.
\newblock {\em Pafnuty Publications}, Oxford 2014.

\bibitem{Beck_1977}
J.~V. Beck and K.~J. Arnold.
\newblock {\em Parameter Estimation in Engineering and Science}.
\newblock John Wiley and Sons, New York, 1977.

\bibitem{Ozisik_2000}
M.~Necati Ozisik and H.~R.B. Orlande.
\newblock {\em Inverse Heat Transfer: Fundamentals and Applications}.
\newblock CRC Press, New York, 2000.

\bibitem{Kabanikhin_2008b}
S.~I. Kabanikhin, A.~Hasanov, and A.V. Penenko.
\newblock A gradient descent method for solving an inverse coefficient heat
  conduction problem.
\newblock {\em Numerical Analysis and Applications}, 1(1):34--45, 2008.

\end{thebibliography}

\appendix

\section{Sensitivity equations}
\label{app:sensitivity_equations}

The sensitivity equations for each direct model are detailed in this section. From these equations, numerical models are built to compute the sensitivity functions using the same approach as for the direct model . More specifically, for the DF model, the numerical scheme for the sensitivity equations is built by differentiating the fully discrete equation~\eqref{eq:DF} to the unknown parameter. Thus, the direct extension of the \DF ~scheme for the sensitivity equations is obtained. For the lumped RC model, the numerical scheme is constructed by differentiating each term of the matrix formulation in equation~\eqref{eq:RC_NM} with respect to to the unknown parameter. So the computation of the sensitivity function is also based on explicit \textsc{Euler} time scheme.

\subsection{Estimation of the volumetric heat capacity}

\subsubsection{The DF model}

We define the sensitivity function of the dimensionless temperature relatively to the dimensionless heat capacity by:
\begin{align*}
\Theta \,:\, \bigl[\, 0 \,,\, 1 \,\bigr] \, \times \, \bigl[\, 0 \,,\, t_{\,\fin}^{\,\star} \,\bigr] & \longrightarrow \, \mathbb{R} \,, \\[4pt]
\bigl(\,x^{\,\star} \,,\, t^{\,\star}\,\bigr) & \longmapsto \, \pd{u}{c^{\,\star}}\bigl(\,x^{\,\star} \,,\, t^{\,\star}\,\bigr) \,.
\end{align*}
It is computed by solving the following differential equation by differentiating Eq.~\eqref{eq:heat_diffusion_dimless} relatively to $c^{\,\star}\,$:
\begin{align*}
c^{\,\star} \cdot \pd{\Theta}{t^{\,\star}} \egal \Fo \cdot k^{\,\star} \cdot \pd{^{\,2} \Theta}{x^{\,\star\,2}}  \moins \pd{u}{t^{\,\star}} \,,
\end{align*}
with the following boundary conditions:
\begin{align*}
k^{\,\star} \cdot \pd{\Theta}{x^{\,\star}} & \egal \Bi \cdot h^{\,\star}_{\,L} \cdot \Theta \,, && x^{\,\star} \egal 0 \,, \\[4pt]
k^{\,\star} \cdot \pd{\Theta}{x^{\,\star}} & \egal - \, \Bi \cdot h^{\,\star}_{\,R} \cdot \Theta \,, && x^{\,\star} \egal 1  \,,
\end{align*}
and the initial condition:
\begin{align*}
\Theta \egal 0 \,.
\end{align*}
The solution of this problem gives the sensitivity of the dimensionless field $u$ with respect to the heat capacity for the DF model.

\subsubsection{The RC ~model}

We define the sensitivity function of the temperature relatively to the volumetric heat capacity:
\begin{align*}
X_{\,j} \,:\, \bigl[\, 0 \,,\, t_{\,\fin} \,\bigr] & \longrightarrow \, \mathbb{R} \,,  && j \, \in \, \bigl\{\, 1 \,,\, 2 \,,\, 3\,\bigr\}  \,, \\[4pt]
t & \longmapsto \, \pd{T_{\,j}}{c}\,\bigl(\,t\,\bigr) \,.
\end{align*}
Three equations are obtained by differentiating Eqs.~\eqref{eq:RC_model} and \eqref{eq:RC_model_BC} with respect to $c\,$:
\begin{align*}
e^{\,2} \cdot  c \cdot \od{X_{\,2}}{t} & \egal k \cdot \biggl(\, X_{\,3} \moins 2 \cdot X_{\,2} \plus X_{\,1} \,\biggr)  \moins e^{\,2} \cdot \od{T_{\,2}}{t} \,, \\[4pt]
\frac{k}{e} \cdot \Bigl(\, X_{\,2} \moins X_{\,1} \,\Bigr) &  \egal h_{\,L} \cdot X_{\,1} \,, \\[4pt]
\frac{k}{e} \cdot \Bigl(\, X_{\,3} \moins X_{\,2} \,\Bigr) &  \egal - \, h_{\,R} \cdot X_{\,3} \,.
\end{align*}
The initial condition is:
\begin{align*}
X_{\,j} \egal 0 \,, \qquad j \, \in \, \bigl\{\, 1 \,,\, 2 \,,\, 3\,\bigr\} \,.
\end{align*}
The solution gives the sensitivity of the temperature with respect to the heat capacity for the RC model.

\subsection{Estimation of the thermal conductivity}

\subsubsection{The DF model}

We define the sensitivity function of the dimensionless temperature relatively to the dimensionless thermal conductivity by:
\begin{align*}
\Theta \,:\, \bigl[\, 0 \,,\, 1 \,\bigr] \, \times \, \bigl[\, 0 \,,\, t_{\,\fin}^{\,\star} \,\bigr] & \longrightarrow \, \mathbb{R} \,, \\[4pt]
\bigl(\,x^{\,\star} \,,\, t^{\,\star}\,\bigr) & \longmapsto \, \pd{u}{k^{\,\star}}\bigl(\,x^{\,\star} \,,\, t^{\,\star}\,\bigr) \,.
\end{align*}
It is computed by solving the following differential equation by differentiating Eq.~\eqref{eq:heat_diffusion_dimless} relatively to $k^{\,\star}\,$:
\begin{align*}
c^{\,\star} \cdot \pd{\Theta}{t^{\,\star}} \egal \Fo \cdot k^{\,\star} \cdot \pd{^{\,2} \Theta}{x^{\,\star \, 2}} \plus \Fo \cdot \pd{^{\,2} u}{x^{\,\star \, 2}} \,,
\end{align*}
with the following boundary conditions:
\begin{align*}
k^{\,\star} \cdot \pd{\Theta}{x^{\,\star}} & \egal \Bi \cdot h^{\,\star}_{\,L} \cdot \Theta \moins \pd{u}{x^{\,\star}} \,, && x^{\,\star} \egal 0 \,, \\[4pt]
k^{\,\star} \cdot \pd{\Theta}{x^{\,\star}} & \egal - \, \Bi \cdot h^{\,\star}_{\,R} \cdot \Theta  \moins \pd{u}{x^{\,\star}}\,, && x^{\,\star} \egal 1 \,, 
\end{align*}
and the initial condition:
\begin{align*}
\Theta \egal 0 \,.
\end{align*}
This problem enables to compute the sensitivity of the field $u$ with respect to the parameter $k^{\,\star}$ for the DF model. 

\subsubsection{The RC ~model}

We define the sensitivity function of the temperature relatively to the thermal conductivity:
\begin{align*}
X_{\,j} \,:\, \bigl[\, 0 \,,\, t_{\,\fin} \,\bigr] & \longrightarrow \, \mathbb{R} \,,  && j \, \in \, \bigl\{\, 1 \,,\, 2 \,,\, 3\,\bigr\}  \,, \\[4pt]
t & \longmapsto \, \pd{T_{\,j}}{k}\bigl(\,t\,\bigr) \,.
\end{align*}
Three equations are obtained by differentiating Eqs.~\eqref{eq:RC_model} and \eqref{eq:RC_model_BC} relatively to $k\,$:
\begin{align*}
e^{\,2} \cdot  c \cdot  \od{X_{\,2}}{t} & \egal k \cdot \biggl(\, X_{\,3} \moins 2 \cdot X_{\,2} \plus X_{\,1} \,\biggr)  \plus \biggl(\, T_{\,3} \moins 2 \cdot T_{\,2} \plus T_{\,1} \,\biggr) \,, \\[4pt]
\frac{k}{e} \cdot \Bigl(\, X_{\,2} \moins X_{\,1} \,\Bigr) &  \egal h_{\,L} \cdot X_{\,1} 
\moins \frac{1}{e} \cdot \Bigl(\, T_{\,2} \moins T_{\,1} \,\Bigr) \,, \\[4pt]
\frac{k}{e} \cdot \Bigl(\, X_{\,3} \moins X_{\,2} \,\Bigr) &  \egal - \, h_{\,R} \cdot X_{\,3} 
 \moins \frac{1}{e} \cdot \Bigl(\, T_{\,3} \moins T_{\,2} \,\Bigr)\,.
\end{align*}
The initial condition is:
\begin{align*}
X_{\,j} \egal 0 \,, \qquad j \, \in \, \bigl\{\, 1 \,,\, 2 \,,\, 3\,\bigr\} \,.
\end{align*}
The solution of this problem gives the sensitivity of temperature with respect to the thermal conductivity in the RC model. 

\subsection{Estimation of the left heat transfer coefficient}

\subsubsection{The DF model}

We define the sensitivity function of the dimensionless temperature relatively to the dimensionless heat transfer coefficient by:
\begin{align*}
\Theta \,:\, \bigl[\, 0 \,,\, 1 \,\bigr] \, \times \, \bigl[\, 0 \,,\, t_{\,\fin}^{\,\star} \,\bigr] & \longrightarrow \, \mathbb{R} \,, \\[4pt]
\bigl(\,x^{\,\star} \,,\, t^{\,\star}\,\bigr) & \longmapsto \, \pd{u}{h^{\,\star}_{\,L}}\bigl(\,x^{\,\star} \,,\, t^{\,\star}\,\bigr) \,.
\end{align*}
It is computed by solving the following differential equation by differentiating Eq.~\eqref{eq:heat_diffusion_dimless} relatively to $h^{\,\star}_{\,L}\,$:
\begin{align*}
c^{\,\star} \cdot \pd{\Theta}{t^{\,\star}} \egal \Fo \cdot k^{\,\star} \cdot \pd{^{\,2} \Theta}{x^{\,\star\,2}} \,,
\end{align*}
with the following boundary conditions:
\begin{align*}
k^{\,\star} \cdot \pd{\Theta}{x^{\,\star}} & \egal \Bi \cdot h^{\,\star}_{\,L} \cdot \Theta \plus \Bi \cdot u \,, && x^{\,\star} \egal 0 \\[4pt]
k^{\,\star} \cdot \pd{\Theta}{x^{\,\star}} & \egal - \, \Bi \cdot h^{\,\star}_{\,R} \cdot \Theta \,, && x^{\,\star} \egal 1 
\end{align*}
and the initial condition:
\begin{align*}
\Theta \egal 0 \,.
\end{align*}
It permits to compute the sensitivity of the field $u$ with respect to the parameter $h^{\,\star}_{\,L}$ for the DF model. 

\subsubsection{The RC ~model}

We define the sensitivity function of the temperature relatively to the heat transfer coefficient:
\begin{align*}
X_{\,j} \,:\, \bigl[\, 0 \,,\, t_{\,\fin} \,\bigr] & \longrightarrow \, \mathbb{R} \,,  && j \, \in \, \bigl\{\, 1 \,,\, 2 \,,\, 3\,\bigr\}  \,, \\[4pt]
t & \longmapsto \, \pd{T_{\,j}}{h_{\,L}}\,\bigl(\,t\,\bigr) \,.
\end{align*}
Three equations are obtained by differentiating Eqs.~\eqref{eq:RC_model} and \eqref{eq:RC_model_BC} relatively to $h_{\,L}\,$:
\begin{align*}
e^{\,2} \cdot  c \, \od{X_{\,2}}{t} & \egal k \cdot \biggl(\, X_{\,3} \moins 2 \cdot X_{\,2} \plus X_{\,1} \,\biggr) \,, \\[4pt]
\frac{k}{e} \cdot \Bigl(\, X_{\,2} \moins X_{\,1} \,\Bigr) &  \egal h_{\,L} \cdot X_{\,1} \plus   T_{\,1}  \moins  T_{\,\infty\,,\,L} \,, \\[4pt]
\frac{k}{e} \cdot \Bigl(\, X_{\,3} \moins X_{\,2} \,\Bigr) &  \egal - \, h_{\,R} \cdot X_{\,3} \,.
\end{align*}
The initial condition is:
\begin{align*}
X_{\,j} \egal 0 \, \qquad j \, \in \, \bigl\{\, 1 \,,\, 2 \,,\, 3\,\bigr\} \,.
\end{align*}
With this model, we compute the sensitivity of temperature with respect to parameter $h_{\,L}$ in the RC model. 

\end{document}